\documentclass[onefignum,onetabnum]{siamart171218}

\usepackage{lipsum}
\usepackage{amsfonts}
\usepackage{graphicx}
\usepackage{epstopdf}
\usepackage{algorithmic}
\usepackage{multirow}
\usepackage{float}
\usepackage{epstopdf}
\usepackage{appendix}
\usepackage{bm}
\usepackage{dsfont}
\usepackage{amsmath}
\usepackage{amsfonts}
\usepackage{amssymb}
\usepackage{subcaption}
\usepackage{caption}
\usepackage{enumitem}
\usepackage{array}
\ifpdf
  \DeclareGraphicsExtensions{.eps,.pdf,.png,.jpg}
\else
  \DeclareGraphicsExtensions{.eps}
\fi

\newsiamremark{remark}{Remark}
\newsiamremark{hypothesis}{Hypothesis}
\newsiamremark{example}{Example}
\crefname{hypothesis}{Hypothesis}{Hypotheses}
\newsiamthm{claim}{Claim}
\newsiamthm{assumption}{Assumption}

\headers{Robust Optimization for Sequential Decision Making}{Z. Chen, P. Yu, AND W. B. Haskell}

\title{Distributionally Robust Optimization for Sequential Decision Making\thanks{Submitted to the editors on \today.
\funding{This work was supported by the Ministry of Education of Singapore through grant MOE2014-T2-1-066.}}}

\author{Zhi Chen\thanks{Imperial College Business School, Imperial College London 
  (\email{zhi.chen@imperial.ac.uk}).}
\and Pengqian Yu\thanks{Neuri Pte Ltd 
  (\email{pengqian@neuri.ai}).}
\and William B. Haskell\thanks{Department of Industrial Systems Engineering and Management, National University of Singapore}
  (\email{isehwb@nus.edu.sg}).}

\usepackage{amsopn}

\newcommand{\bmt}[1]{\tilde{\bm{#1}}}

\newcommand{\bmb}[1]{\bar{\bm{#1}}}

\def\n{\tilde{n}}

\newcommand{\ep}[1]{\mathbb{E}_{\mathbb{P}}\left[#1\right]}

\newcommand{\epbar}[1]{\mathbb{E}_{\bar{\mathbb{P}}}\left[#1\right]}

\newcommand{\epn}[1]{\mathbb{E}_{\mathbb{P}_n}\left[#1\right]}

\newcommand{\pp}[1]{\mathbb{P}\left[#1\right]}

\newcommand{\ppn}[1]{\mathbb{P}_n\left[#1\right]}

\DeclareMathOperator*{\arginf}{\arg\!\inf}
\DeclareMathOperator*{\argmax}{\arg\!\max}

\ifpdf
\hypersetup{
  pdftitle={Distributionally Robust Optimization for Sequential Decision Making},
  pdfauthor={Z. Chen, P. Yu, AND W. B. Haskell}
}
\fi

\begin{document}

\maketitle

\begin{abstract}
The distributionally robust Markov Decision Process (MDP) approach asks for a distributionally robust policy that achieves the maximal expected total reward under the most adversarial distribution of uncertain parameters. In this paper, we study distributionally robust MDPs where ambiguity sets for the uncertain parameters are of a format that can easily incorporate in its description the uncertainty's generalized moment as well as statistical distance information. In this way, we generalize existing works on distributionally robust MDP with generalized-moment-based and statistical-distance-based ambiguity sets to incorporate information from the former class such as moments and dispersions to the latter class that critically depends on empirical observations of the uncertain parameters. We show that, under this format of ambiguity sets, the resulting distributionally robust MDP remains tractable under mild technical conditions. To be more specific, a distributionally robust policy can be constructed by solving a sequence of one-stage convex optimization subproblems.
\end{abstract}



\section{Introduction}

Sequential decision making in stochastic dynamic environments, also called the ``planning problem", is often modeled using a Markov Decision Process (MDP, cf. \cite{Bertsekas_2000, Puterman_2014}). A strategy that achieves maximal expected total reward is considered optimal.  In practice, parameter uncertainty---the deviation of the model parameters (rewards and transition probabilities)  from the true ones---often causes the performance of ``optimal'' policies to degrade significantly (see experiments in \cite{Mannor_2007}). Many efforts have been made to reduce such performance variation under the robust MDP framework (e.g., \cite{Iyengar_2005, Nilim_Ghaoui_2005, White_1994, Wiesemann_Kuhn_Rustem_2013}). In this context, it is assumed that the parameters can be any member of a known set---termed the uncertainty set---and solutions are ranked based on their performance under the worst (i.e., most adversarial) parameter realizations.

New models on parameter uncertainty have been inspired by recent advances in the emerging field of distributionally robust optimization, which seeks for a maximal worst-case expected performance with reference to an ambiguity set---a family of distributions that identically share the given a priori information on the uncertainty. A typical approach for constructing the ambiguity set specifies support information and generalized moment conditions on parameter uncertainty (see e.g., \cite{Delage_Ye_2010, Goh_Sim_2010, Wiesemann_Kuhn_Sim_2014}, and references therein). For example, pioneer works in robust MDPs (see e.g., \cite{Iyengar_2005, Nilim_Ghaoui_2005, Wiesemann_Kuhn_Rustem_2013}) consider parameter uncertainty such that only their support is known. The work \cite{Xu_Mannor_2012} studies parameter uncertainty that lies in a collection of nested subsets of the support subject to different confidence levels. More recently, distributionally robust MDPs where the ambiguity set of the parameter uncertainty is a form proposed by \cite{Wiesemann_Kuhn_Sim_2014} have been examined in \cite{Yu_Xu_2016}.

Another approach for constructing the ambiguity set has also attracted considerable interest, which characterizes distributions by their proximity to a reference distribution via certain statistical distance measure. By varying the statistical distance measure, examples include the $\phi$-divergence ambiguity set (see e.g., \cite{Ben_Hertog_Waegenaere_Melenberg_Rennen_2013, Jiang_Guan_2016, Wang_Glynn_Ye_2016}) and Wasserstein ambiguity set (see e.g., \cite{Gao_Kleywegt_2016, Esfahani_Kuhn_2015, Pflug_Wozabal_2007, Zhao_Guan_2015}).  A soft-constrained version of distributionally robust MDPs with a particular $\phi$-divergence ambiguity set (to be more precise, the Kullback-Leibler-divergence ambiguity set) has been studied in \cite{Osogami_2012}, and it is shown to be equivalent to a risk-sensitive MDP that considers the expected exponential utility. However, except that, distributionally robust MDPs with other $\phi$-divergence ambiguity sets---for example, the likelihood robust ambiguity set in \cite{Wang_Glynn_Ye_2016}---have not been studied. Distributionally robust MDPs where the Wasserstein distance is used to specify an ambiguity set have appeared in \cite{Yang_2017}. 

The two mentioned approaches are usually separately studied in the literature of distributionally robust optimization. Each of them has their own advantages. For instance, a larger set of historical data will provide better estimations about inputs for a generalized-moment-based ambiguity set without increasing the size of the corresponding distributionally robust optimization problem, and a statistical-distance-based ambiguity set could yield solutions that have good out-of-sample performance in terms of variability and disappointment (see e.g., \cite{Gotoh_Kim_Lim_2015, Lam_2016, Esfahani_Kuhn_2015, Van_Esfahani_Kuhn_2017}). Recently, \cite{Chen_Sim_Xiong_2017} propose a generic format of ambiguity sets for distributionally robust optimization models. That format can express ambiguity sets constructed from the two approaches in a unified manner. Inspired by that progress, we study distributionally robust MDPs with respect to ambiguity sets in such a format. Our work closely relates to \cite{Yu_Xu_2016} for distributionally robust MDPs. However, we generalize the results therein in a sense that there are several new examples (including the popular Wasserstein ambiguity set as well as a new class of ambiguity sets based on mixture distributions) in the format of ambiguity set we consider. We summarize the contributions of our paper as follows.
\begin{enumerate}
\item
We study distributionally robust MDPs with a general format of ambiguity sets that can express the two notable classes of ambiguity sets that are respectively based on generalized moment and statistical distance in a unified manner. In this way, we generalize existing works on distributionally robust MDPs that usually investigate these two classes separately. 
\item 
A notable result following from the unified format is that we are able to obtain the $\mathcal{S}$-robust strategy for finite-stage distributionally robust MDPs with the Wasserstein ambiguity set. Our approach solves classical robust optimization problems and differs from \cite{Yang_2017} that solves a general convex program for the $\mathcal{S}$-robust strategy. In this regard, our approach may be more friendly to practitioners because classical robust optimization problems nowadays can be specified in an efficient yet intuitive way by using algebraic modeling packages, which will automatically derive compact mathematical reformulations of robust optimization problems and pass them to the off-the-shelf commercial solvers. 
\item 
Our analyses naturally apply to distributionally robust MDPs with respect to tailored ambiguity sets that are a hybridization of a generalized-moment-based ambiguity set and a statistical-distance-based ambiguity set, such as a hybridization of the ambiguity set in \cite{Yu_Xu_2016} and that in \cite{Yang_2017} and a hybridization of the mean-covariance ambiguity and a likelihood robust ambiguity set (as suggested in \cite{Wang_Glynn_Ye_2016}). The hybridization can leverage the benefits from both ambiguity sets by encoding structure information, which is typically modeled by generalized moment constraints, in statistical-distance-based ambiguity sets that largely depend on data and hence would otherwise ignore any prior knowledge about structure information. To the best of our knowledge, this type of ambiguity sets has not been considered in distributionally robust MDPs. We hope our framework would encourage further researches along this line.
\end{enumerate}

This paper is organized as follows. In \cref{section2} we provide backgrounds and assumptions on classical and distributionally robust MDPs and we present several examples of ambiguity sets in the format we study in this paper. We formulate and solve distributionally robust MDPs for the finite horizon case in \cref{section3} and we expand our analysis to the infinite horizon case in \cref{section5}. In \cref{section6}, we conduct numerical experiments to verify the validity and effectiveness of our proposed approach. Some concluding remarks are offered in \cref{section7}.

\noindent \textbf{Notations.}
Throughout the paper, we use boldface uppercase and lowercase characters to denote (respectively) matrices and vectors.
Special vectors of the appropriate dimension include $\bm{0}$, $\bm{e}$, and $\bm{e}_i$, which respectively correspond to the vector of $0$s, the vector of $1$s, and the $i^{\rm th}$ unit standard basis. 
We denote by $ [N] = \{1,2,\ldots,N\} $ the set of positive running indices up to $ N $. 
We denote the set of all Borel probability distributions on a set $ \mathcal{A} \in \mathbb{R}^N $ by $ \mathcal{P}_0(\mathcal{A}) $.  
A random vector $ \bmt{z} $ is denoted with a tilde sign and we use $ \bmt{z} \sim \mathbb{P} $ to denote that $\bmt{z} $ is governed by a  probability distribution $ \mathbb{P} $. Given a probability distribution $ \mathbb{P} $, we use $ \ep{\cdot} $ and $ \pp{\cdot} $ to denote the corresponding expectation and probability. We say a convex set is tractable if its set membership can be described by finitely many convex constraints and, potentially, auxiliary variables. Similarly, a convex function is tractable if its epigraph is.

\section{Preliminaries}\label{section2}
In this section, we provide backgrounds and assumptions on classical and distributionally robust MDPs. We also discuss in details ambiguity sets for the uncertain parameters---one of the most key ingredients in distributionally robust MDPs.
\subsection{Classical Markov Decision Processes}
A (finite) MDP is defined as a 6-tuple $\langle \mathcal{S},\mathcal{A},T,\gamma,\bm{p},\bm{r}\rangle$. Both the state space $\mathcal{S}$ and the action space $\mathcal{A}$ are finite. The decision horizon $T$ is possibly infinite and the discount factor $\gamma\in(0,1)$. Given the planning horizon $T$, the state and action at time $t=1,\dots,T$ are denoted by $s_t$ and $a_t$, respectively. The parameter $\bm{p}$ and  $\bm{r}$ are the transition probability and the expected reward, i.e., for a state $s$ and an action $a$, $r(s,a)$ is the nonnegative and bounded expected reward and $p(s'|s,a)$ is the probability of visiting the next state $s'$. We use subscript $s$ to denote the value associated with the state $s$, e.g., $\mathcal{A}_s$ is the set of actions in state $s$ and $\mathcal{A} \triangleq \bigcup_{s\in\mathcal{S}}\mathcal{A}_s$, $\bm{r}_s\triangleq (r(s,a))_{a \in \mathcal{A}_s}$ denotes the vector form of the rewards associated with the state $s$, $\bm{\pi}_s \triangleq (\pi_s(a))_{a \in \mathcal{A}_s} \in  \mathcal{P}(\mathcal{A}_s) $ specifies the probabilities that the actions chosen at state $s$ for a strategy $\bm{\pi}$. The elements in the vector $\bm{p}_s\triangleq (p(s'|s,a))_{s'\in\mathcal{S},a\in \mathcal{A}_s}$ are listed as follows: the transition probabilities of the same action are arranged in the same block and inside each block they are listed according to the order of the next state. Finally, we have $\bm{p} = (\bm{p}_s)_{s \in \mathcal{S}}$ and $\bm{r} = (\bm{r}_s)_{s\in \mathcal{S}}$.

Let $\mathcal{H}_t$ be the set of histories at time $t$, given by $\mathcal{H}
_1\triangleq\mathcal{S}$ and $\mathcal{H}_t\triangleq(\mathcal{S}\times\mathcal{A})^t\times\mathcal{S}$ for all $t\geq2 $.
A history dependent randomized decision rule is a mapping $d_t:\mathcal{H}_t \mapsto  \mathcal{P}(\mathcal{A})$, where $  \mathcal{P}(\mathcal{A}) $ is the probability simplex on the set of actions $ \mathcal{A} $. A Markovian randomized decision rule is a mapping $d_t:\mathcal{S} \mapsto  \mathcal{P}(\mathcal{A})$. A deterministic decision rule $d_t:\mathcal{S}\mapsto\mathcal{A}$ can be regarded as a special case of a randomized decision rule in which the probability distribution on the set of actions is degenerate. The sets of history dependent randomized decision rules, Markovian randomized decision rules and Markovian deterministic decision rules are denoted by $ \mathcal{R}^\text{HR}_t $, $ \mathcal{R}^\text{MR}_t $ and $ \mathcal{R}^\text{MD}_t $, respectively. A policy, or a strategy is a sequence of decision rules for the entire time horizon, i.e., $\bm{\pi}=(d_1,d_2,\dots,d_{T-1})$ where $d_t\in \mathcal{R}_t^\text{K}$ and K designates a class of decision rules (K=HR, MR, MD). We denote the set of all policies of class K by $\Pi^\text{K} \triangleq \mathcal{R}_1^\text{K}\times \mathcal{R}_2^\text{K}\times\cdots \times\mathcal{R}_{T-1}^\text{K}$.

Given an initial state $s_1\in\mathcal{S}$, a policy $\bm{\pi} \in \Pi^\text{K}$ determines a probability measure $\mathbb{Q}({\bm{\pi}})$ on the canonical measurable space of MDP trajectories and a corresponding stochastic process $\{(s_t,a_t)\}_{t\geq1}$. The expectation operator with respect to $\mathbb{Q}({\bm{\pi}})$ is denoted by $\mathbb{E}_{\mathbb{Q}(\bm{\pi})}$, and the expected (discounted) total-reward under parameters pair $(\bm{p}, \bm{r})$ can then be expressed as
$$
u(\bm{\pi},\bm{p},\bm{r},s_1)\triangleq \mathbb{E}_{\mathbb{Q}(\bm{\pi})}\left[\sum_{t=1}^{T}\gamma^{t-1}r(s_t,a_t)\right].
$$
The goal of the classical MDPs is to find an optimal strategy that solves the problem
\begin{equation}\label{classical MDP}
\max_{\bm{\pi} \in \Pi^{\text{HR}}} u(\bm{\pi},\bm{p},\bm{r},s_1);
\end{equation}
here, the parameters $ \bm{p} $ and $ \bm{r} $ are fixed and are known to the decision maker. It is well known that problem~\cref{classical MDP} has an optimal policy in $\Pi^{\text{MD}}$ and can be solved by value iteration, policy iteration, and linear programming (cf. \cite{Puterman_2014}). 

\subsection{Distributionally Robust Markov Decision Processes}
A distributionally robust MDP is defined as a tuple $\langle \mathcal{S}, \mathcal{A}, T, \gamma, \mathcal{F}_{\mathcal{S}}\rangle$. In  distributionally robust MDPs, the transition probability $\bmt{p}$ and the expected reward $\bmt{r}$ are unknown. Instead, they are assumed to obey a joint probability distribution $\mathbb{P}$, which is also unknown but belongs to a known family $\mathcal{F}_{\mathcal{S}}$ of probability distributions called the ambiguity set whose members share some identical distributional information. Given the ambiguity set $ \mathcal{F}_{\mathcal{S}} $ and an initial state $s_1 \in \mathcal{S}$, the distributionally robust MDP solves the following problem 
\begin{equation*}
\max_{\bm{\pi} \in\Pi^{\text{HR}}}\inf_{\mathbb{P} \in \mathcal{F_S}}\mathbb{E}_{\mathbb{P}}\left[u(\bm{\pi}, \bmt{p},\bmt{r},s_1)\right]
\end{equation*}
and asks for a policy $\bm{\pi}$ that attains the best worst-case expected performance under an ambiguous probability distribution $\mathbb{P}$ of $ (\bmt{p}, \bmt{r})$ arising from the set $\mathcal{F}_\mathcal{S}$.

While the distributionally robust MDP framework can be very general, not every ambiguity set $ \mathcal{F}_{\mathcal{S}} $ would result in reformulations that are easy to solve. Therefore, in this paper, we focus on a class of ambiguity sets that are general to encompass several types of ambiguity sets that have attracted considerable interest in the literature, and at same time, can guarantee the computational tractability of the arising distributionally robust MDP.

For the first key requirement of $\mathcal{F}_\mathcal{S}$, we adopt the following set of distributions in our model.
$$
\mathcal{F}_\mathcal{S}\triangleq\left\{\mathbb{P} ~\left\vert~ \mathbb{P} = \bigotimes_{s\in \mathcal{S}}\mathbb{P}_{s}, ~\mathbb{P}_s\in \mathcal{F}_s ~~\forall s\in \mathcal{S} \right.\right\}.
$$
Here for each state $s \in \mathcal{S}$, the set $\mathcal{F}_s$ is a family of probability distributions of uncertain parameters $(\bm{p}_s, \bm{r}_s)$ associated with the state $s$, and it is termed as the ``state-wise ambiguity set". Note that $ \bigotimes_{s\in \mathcal{S}}\mathbb{P}_{s}$ stands for the product measure generated by $\mathbb{P}_{s}$, which indicates that the uncertain parameters across different states are independent. Such a state-wise property is called ``s-rectangularity" in the literature (e.g., \cite{Nilim_Ghaoui_2005,Wiesemann_Kuhn_Rustem_2013}) and it plays an essential role in reducing the distributionally robust MDP to a robust MDP (see \cite{Xu_Mannor_2012}).

As the second key requirement of $ \mathcal{F}_\mathcal{S} $, we assume that the admissible state-wise ambiguity set $ \mathcal{F}_s $ for all $ s \in \mathcal{S} $ can be expressed as the union of marginal distributions of the uncertain parameters $ (\bmt{p}_s, \bmt{r}_s) $ under all joint distributions of $ (\bmt{p}_s, \bmt{r}_s, \tilde{n}_s) $. Here $\tilde{n}_s$ is a one-dimensional auxiliary random variable that arises from a lifted ambiguity set $\mathcal{G}_s$---a family of joint distributions of $(\bmt{p}_s, \bmt{r}_s, \n_s)$, whose projection onto $(\bmt{p}_s, \bmt{r}_s)$ gives $ \mathcal{F}_s$. In other words, there exists $ \mathcal{G}_s $ that satisfies 
$
\mathcal{F}_s = \bigcup_{\mathbb{P} \in \mathcal{G}_s} \{\prod_{(\bmt{p}_s, \bmt{r}_s)}\mathbb{P} \},
$ 
where $\mathbb{P} \in \mathcal{G}_s$ is a joint distribution of $(\bmt{p}_s, \bmt{r}_s, \tilde{n}_s)$ and we denote by $\prod_{(\bmt{p}_s, \bmt{r}_s)} \mathbb{P}$ the marginal distribution of $ (\bmt{p}_s, \bmt{r}_s) $ under $ \mathbb{P} $. In particular, we assume that the lifted ambiguity set $ \mathcal{G}_s$ is representable in the format
\begin{equation}
\label{set:general format}
\begin{array}{l}
\mathcal{G}_s \triangleq \\
\left\{\mathbb{P} \in \mathcal{P}_0(\mathbb{R}^{I_{s1}} \times \mathbb{R}^{I_{s2}} \times [N_s]) ~\left\vert~
\begin{array}{ll}
(\bmt{p}_s, \bmt{r}_s, \tilde{n}_s) \sim \mathbb{P} \\
\ep{(\bmt{p}_s, \bmt{r}_s) \mid \tilde{n}_s \in \mathcal{N}_j} = \bm{\mu}_j &~\forall j \in [J_s] \\
\ep{\bm{g}_{j\tilde{n}_s}(\bmt{p}_s, \bmt{r}_s) \mid \tilde{n}_s \in \mathcal{N}_j} \leq \bm{\nu}_j &~\forall j \in [J_s] \\
\pp{(\bmt{p}_s, \bmt{r}_s) \in \mathcal{D}_{n} \mid \tilde{n}_s = n} = 1 &~\forall n \in [N_s] \\
\pp{\tilde{n}_s = n} = \omega_n &~\forall n \in [N_s] \\
{\rm for~ some~} \bm{\omega} \in \mathcal{W}, ~(\bm{\mu}_j, \bm{\nu}_j) \in \mathcal{U}_{j} &~\forall j \in [J_s] 
\end{array}
\right.\right\}.
\end{array}
\end{equation}
Here, dimensions of the uncertain parameters are $I_{s1} = \vert \mathcal{S} \vert$ and $I_{s2} = \vert \mathcal{A}_s \vert$;  for each $j \in [J_s]$, the set $\mathcal{N}_{j} \subseteq [N_s]$ is a subset of scenarios;  for all $n \in \mathcal{N}_j$, $j \in [J_s]$, the function $\bm{g}_{jn}: \mathbb{R}^{I_{s1} + I_{s2}} \mapsto \mathbb{R}^{M_j}$ is convex lower semi-continuous; sets $\mathcal{D}_n \subseteq \mathbb{R}^{I_{s1} + I_{s2}}$, $n \in [N_s]$ and $\mathcal{U}_j \subseteq \mathbb{R}^{I_{s1} + I_{s2} + M_{j}}$, $j \in [J_s]$ are closed, convex and compact; and the set $\mathcal{W} \subseteq {\rm int}\{\bm{w} \in \mathbb{R}^{N_s} \mid \bm{e}^\top\bm{w} = 1\}$. For all $n \in [N_s]$, we also define $\mathcal{J}_n \triangleq \{j \in [J_s] \mid n \in \mathcal{N}_j\}$ for notational convenience. 

One could verify that $\mathcal{G}_s$ is a convex set of joint probability distributions of $(\bmt{p}_s, \bmt{r}_s)$ and $\n_s$. The one-dimensional auxiliary random variable $\n_s$ is discrete with its finitely many scenarios collectively denoted by a set $[N_s]$ and it plays a key role in the format~\cref{set:general format}. Firstly, any $j \in [J_s]$ induces a condition that consists of a subset $\mathcal{N}_j \subseteq [N_s]$ of scenarios and gives the conditional generalized moment information about the uncertain parameters (see the first and second constraint groups).  Secondly, the auxiliary random variable $\n_s$ takes the $n^{\rm th}$ scenario with probability $\omega_n$, and conditioning on this scenario, the support of the uncertainty parameters could be scenario-wise different (see the third constraint group). Finally, the conditional generalize moment information as well as the probabilities of scenarios can all be uncertain and can only be known to reside in some uncertainty sets (see the constraint group in the last line). As we will present in the coming examples, by simply introducing the one-dimensional random variable $\n_s$, the lifted ambiguity set $\mathcal{G}_s$ is intuitive by expressing a rich family of ambiguity sets in the literature and can facilitate the construction of new tailored ambiguity sets.  

To obtain a tractable reformulation, we utilize the concept of conic representation and make the following assumption.
\begin{assumption}\label{assu_slater}
For all $ s \in \mathcal{S} $, the conic representation of the following system 
\begin{equation}
\label{eq:slater condition}
\left\{
\begin{array}{ll}
\displaystyle \sum_{n \in \mathcal{N}_j} \bm{\xi}_{n} = \bm{\phi}_j &~\forall j \in [J_s] \\[3mm]
\displaystyle \sum_{n \in \mathcal{N}_j} \bm{\zeta}_{jn} \leq \bm{\varphi}_j &~\forall j \in [J_s] \\[3mm]
\dfrac{\bm{\zeta}_{jn}}{\tau_n} \geq \bm{g}_{jn}\bigg(\dfrac{\bm{\xi}_n}{\tau_n}\bigg) &~\forall n \in [N_s], ~j \in \mathcal{J}_n \\[3mm]
\dfrac{\bm{\xi}_n}{\tau_n} \in \mathcal{D}_n &~\forall n \in [N_s]\\[3mm]
\dfrac{(\bm{\phi}_j, \bm{\varphi}_j)}{\sum_{i \in \mathcal{N}_j} \tau_i} \in \mathcal{U}_j &~\forall j \in [J_s] \\[3mm]
\bm{\tau} \in \mathcal{W}
\end{array}
\right.
\end{equation}
satisfies the Slater's condition, that is, the conic representation of~\cref{eq:slater condition} has a relatively interior point (see, Theorem 1.4.2 in \cite{Ben-tal_Nemirovski_book}).
\end{assumption}
The Slater's condition of the system~\cref{eq:slater condition}, although it appears to be weakly connected with the format~\cref{set:general format}, is a typical assumption made in the distributionally robust optimization literature. More importantly, as we will show shortly, it guarantees the infimum over the ambiguity set, $\inf_{\mathbb{P} \in \mathcal{F_S}}\mathbb{E}_{\mathbb{P}}\left[u(\bm{\pi}, \bmt{p},\bmt{r},s_1)\right]$ to be attainable.

Using the lifted ambiguity sets $ \mathcal{G}_s, s \in \mathcal{S} $, we define a lifted ambiguity set corresponding to $ \mathcal{F}_\mathcal{S}$. 
\begin{equation}\label{set:decomposite}
\mathcal{G}_\mathcal{S}\triangleq\left\{\mathbb{P} ~\left\vert~ \mathbb{P} = \bigotimes_{s\in \mathcal{S}}\mathbb{P}_{s}, ~\mathbb{P}_s\in \mathcal{G}_s ~\forall s\in \mathcal{S} \right.\right\}.
\end{equation}
By the definition of $\mathcal{F}_\mathcal{S} $ and the relation between $\mathcal{F}_s $ and $\mathcal{G}_s$, the state-wise property applies to $\mathcal{G}_\mathcal{S}$ as well. Moreover, we have
$ \prod_{(\bmt{p}, \bmt{r})}\mathcal{G}_{\mathcal{S}} = \mathcal{F}_\mathcal{S}$,
which implies
$$
\max_{\bm{\pi} \in\Pi^{\text{HR}}}\inf_{\mathbb{P} \in \mathcal{F}_\mathcal{S}}\mathbb{E}_{\mathbb{P}}\left[u(\bm{\pi}, \bmt{p},\bmt{r},s_1)\right] = \max_{\bm{\pi} \in\Pi^{\text{HR}}}\inf_{\mathbb{P} \in \mathcal{G}_\mathcal{S}}w(\bm{\pi},\mathbb{P},s_1)
$$
for $w(\bm{\pi}, \mathbb{P}, s_1)\triangleq \mathbb{E}_{\mathbb{P}}\left[u(\bm{\pi}, \bmt{p},\bmt{r},s_1)\right]$ being the expected performance of a policy $ \bm{\pi} $ under a joint probability distribution $\mathbb{P}$ of $ (\bmt{p}, \bmt{r})$ and auxiliary random variables $(\n_s)_{s \in \mathcal{S}}$. The equality holds because the term $\mathbb{E}_{\mathbb{P}}\left[u(\bm{\pi}, \bmt{p},\bmt{r},s_1)\right]$ does not involve those auxiliary variables $(\n_s)_{s \in \mathcal{S}}$ in $ \mathcal{G}_\mathcal{S}$. Thus for the rest of this paper, we will use the lifted ambiguity set $\mathcal{G}_\mathcal{S}$ and work on 
$$
\max_{\bm{\pi} \in\Pi^{\text{HR}}}\inf_{\mathbb{P} \in \mathcal{G_S}}w(\bm{\pi},\mathbb{P},s_1). 
$$

\subsection{Examples of the Lifted Ambiguity Set $\mathcal{G}_s$}\label{section4}
Despite its apparent simplicity, the abstract format~\cref{set:general format} allows us to recover two important and distinct classes of ambiguity sets that are respectively based on generalized moment and statistical distance in a unified format. By doing so, we are also able to construct a hybridization of these two classes of ambiguity sets, which to the best of our knowledge has not been considered in distributionally robust MDPs, but we believe to be potential to leverage the generalized moment as well as statistical distance information at the same time. Besides, we are able to consider a class of mixture-distribution-based ambiguity sets whose member distributions encompass the mixture Gaussian distributions that have been used in MDPs (see e.g., \cite{Hoffman_Freitas_Doucet_Petrers2009,Yu_Xu_2016}). In the remaining of this section, we provide some concrete examples of the lifted ambiguity set $\mathcal{G}_s$ in the format~\cref{set:general format}.

\subsection*{Generalized-moment-based ambiguity sets}
Generalized-moment-based a- 

\noindent mbiguity sets involve probability constraints on the support and expectation constraints on the generalized moment of the ambiguous distributions.

\begin{example}[Support]
We consider the following ambiguity set with $N_s = 1$ and $J_s = 1$, which only contains the support information about the uncertainty: 
\begin{equation*}
\mathcal{G}_s = \left\{\mathbb{P} \in \mathcal{P}_0(\mathbb{R}^{I_{s1}} \times \mathbb{R}^{I_{s2}} \times \{1\}) ~\left\vert~
\begin{array}{ll}
(\bmt{p}_s, \bmt{r}_s, \n_s) \sim \mathbb{P} \\
\pp{(\bmt{p}_s, \bmt{r}_s) \in \mathcal{D} \mid \n_s = 1} = 1 \\
\pp{\n_s = 1} = 1
\end{array}
\right.\right\},
\end{equation*}
where the support set $ \mathcal{D} $ is tractable. In this case, distributionally robust MDPs recover classical robust MDPs (e.g., \cite{Iyengar_2005, Nilim_Ghaoui_2005}) that consider the worst-case scenario. It is well known that, due to Fenchel duality (see Proposition 3.1 in \cite{Follmer_Schied_2002} or Theorem 2.2 in \cite{Ruszczynski_Shapiro_2006}), any distributionally robust objective $\max_{\mathbb{P} \in \mathcal{G}_s}\ep{\tilde{r}} $, an expectation with respect to the worst-case density function $\mathbb{P}$ chosen adversarially from the ambiguity set $ \mathcal{G}_s $, is equivalent to a coherent risk measure that satisfies convexity, monotonicity, translation equivariance and positive homogeneity (see \cite{Artzner_1999}). When the ambiguity set only contains the support information, several classes of support sets would correspond to some popular coherent risk measures in finance (see \cite{Bertsimas_Brown_2009} and \cite{Natarajan_Pachamanova_Sim_2009}). We remark that real-world uncertainty originates from modeling errors (model uncertainty), and perhaps more importantly, from stochastic dynamics. A prudent policy should protect against both types of uncertainties. The Fenchel duality of coherent risk measures naturally relates the risk to model uncertainty. For risk in MDPs, a similar connection was made with time-consistent Markov coherent risk measures \cite{Osogami_2012}. Therefore, by carefully shaping the risk-criterion, it would enable decision makers to consider uncertainty in a broad sense. Designing a principled procedure for such risk-shaping is not trivial and is beyond the scope of this paper. However, we believe that there is much potential to risk-shaping as it may be the key for handling model misspecification in dynamic decision-making.
\end{example}

\begin{example}[Uncertain mean]
Suppose beyond the support $ \mathcal{D}$, we further have some knowledge about the uncertain mean $\ep{(\bmt{p}_s,\bmt{r}_s)}$ of the uncertainty, for instances, $\ep{(\bmt{p}_s,\bmt{r}_s)}  \in  [\underline{\bm{\mu}},\overline{\bm{\mu}}]$ and  $\| \ep{(\bmt{p}_s,\bmt{r}_s)} - \bm{\mu}^0 \|_2 \leq \theta $, where $ \bm{\mu}^0 $ is the estimation of the true mean. We can specify an ambiguity set in the general format~\cref{set:general format} as follows:
\begin{equation*}
\mathcal{G}_s = \left\{\mathbb{P} \in \mathcal{P}_0(\mathbb{R}^{I_{s1}} \times \mathbb{R}^{I_{s2}} \times \{1\}) ~\left\vert~
\begin{array}{ll}
(\bmt{p}_s, \bmt{r}_s, \n_s) \sim \mathbb{P} \\
\ep{(\bmt{p}_s, \bmt{r}_s) \mid \n_s = 1} = \bm{\mu} \\
\pp{(\bmt{p}_s, \bmt{r}_s) \in \mathcal{D} \mid \n_s = 1} = 1 \\
\pp{\n_s = 1} = 1 \\
{\rm for~ some~} \bm{\mu} \in \mathcal{U} 
\end{array}
\right.\right\},
\end{equation*}
where $ \mathcal{U} = \{\bm{\mu} \mid \underline{\bm{\mu}} \leq \bm{\mu} \leq \overline{\bm{\mu}}, ~\| \bm{\mu} - \bm{\mu}^0 \|_2 \leq \theta \} $. 
\end{example}

\subsection*{Statistical-distance-based ambiguity sets}
In the literature of distributionally robust optimization, the class of statistical-distance-based ambiguity sets has also attracted a great deal of interest. In a typical statistical-distance-based ambiguity set, ambiguous distributions are characterized by their proximity to a reference distribution through some statistical distance measures, including $ \phi$-divergences and the Wasserstein distance. The reference distribution is usually chosen to be the empirical distribution
${\mathbb{P}}^{\dag} = \frac{1}{N} \sum_{n \in [N]} \delta_{(\bm{p}^\dag_n, \bm{r}^\dag_n)}$, where $\delta$ is the Dirac distribution and $\{(\bm{p}^\dag_n, \bm{r}^\dag_n)\}_{n \in [N]}$ are empirical observations of $(\bmt{p}_s, \bmt{r}_s)$. Hence the statical-distance-based ambiguity sets are known to be favorable for a data-driven setting.

\begin{example}[$\phi$-divergence ambiguity set]
Let $ J_s = N_s = N $ and the support set $\mathcal{D}_{n} = \left\{(\bm{p}^\dag_n, \bm{r}^\dag_n)\right\}$ being a singleton for all $n \in [N]$. The $\phi$-divergence ambiguity set is defined as follows:
$$
\begin{array}{l}
\mathcal{G}_\phi(\theta) = \\
\left\{\mathbb{P} \in \mathcal{P}_0(\mathbb{R}^{I_{s1}} \times \mathbb{R}^{I_{s2}} \times [N]) ~\left\vert~
\begin{array}{ll}
(\bmt{p}_s, \bmt{r}_s, \n_s) \sim \mathbb{P} \\
\pp{(\bmt{p}_s, \bmt{r}_s) \in {\mathcal D}_n \mid \n_s = n} = 1 ~~\forall n \in [N] \\
\pp{\n_s = n} = \omega_n ~~\forall n \in [N] \\
\mbox{for some } \bm{\omega} \in \mathcal{W}(\theta)
\end{array}\right.
\right\},
\end{array}
$$
where 
$\mathcal{W}(\theta) = \{\bm{\omega} \in \mathbb{R}^{N}_+ \mid \sum_{j\in[J]} \omega_j = 1, \sum_{j \in [J]} \frac{1}{N} \phi (N\omega_j) \leq \theta\}$
with some divergence function $ \phi(\cdot) $ that is convex on $ \mathbb{R}_+ $ and satisfies $ \phi(1) = 0 $, $ 0 \phi(t/0) \triangleq t \lim_{a \to \infty} \phi(a)/a $ for $ t > 0 $, and $ 0 \phi(0/0) = 0 $. We refer to Table~2 in \cite{Ben_Hertog_Waegenaere_Melenberg_Rennen_2013} for a comprehensive list of $\phi$-divergence functions.
\end{example}

\begin{example}[Wasserstein ambiguity set]
Given a tractable distance metric $ \rho:\mathbb{R}^{I_{s1}+I_{s2}} \times \mathbb{R}^{I_{s1}+I_{s2}} \mapsto \mathbb{R}_+ $ such that $\rho((\bm{p}_s, \bm{r}_s), (\bm{p}^\dag, \bm{r}^\dag)) = 0 $ if and only if $ (\bm{p}_s, \bm{r}_s) = (\bm{p}^\dag, \bm{r}^\dag) $, the type-$1$ Wasserstein distance is defined via
\begin{equation*}
\nonumber
\label{Wasserstein distance}
\begin{array}{rcl}
d_W(\mathbb{P},\mathbb{P}^\dag) \triangleq &\inf 
& \epbar{\rho((\bmt{p}_s, \bmt{r}_s), (\bmt{p}^\dag, \bmt{r}^\dag))} \\
&{\rm s.t.} &(\bmt{p}_s, \bmt{r}_s) \sim \mathbb{P}, ~(\bmt{p}^\dag, \bmt{r}^\dag) \in \mathbb{P}^\dag \\
&& \Pi_{(\bmt{p}_s, \bmt{r}_s)}\bar{\mathbb{P}} = \mathbb{P} \\
&& \Pi_{(\bmt{p}^\dag, \bmt{r}^\dag)}\bar{\mathbb{P}} = \mathbb{P}^\dag. 
\end{array}
\end{equation*}
Note that the distance metric is usually chosen as a $p$-norm for some real number $p \geq 1$. The Wasserstein ambiguity set with a support set $ \mathcal{D} $ is defined as
\begin{equation*}
\label{set:statistical-distance-based Wasserstein}
\mathcal{F}_W(\theta) \triangleq 
\left\{
\mathbb{P} \in \mathcal{P}_0(\mathcal{D})~\left\vert~
\begin{array}{ll}
(\bmt{p}_s, \bmt{r}_s) \sim \mathbb{P}, ~(\bmt{p}^\dag,\bmt{r}^\dag) \sim \mathbb{P}^\dag \\
d_W(\mathbb{P}, \mathbb{P}^\dag) \leq \theta \\
\end{array} \\
\right.\right\},
\end{equation*}
which is a Wasserstein ball of radius $ \theta $ centered around the empirical probability distribution ${\mathbb{P}}^{\dag}$. It is well known that the Wasserstein approach \emph{(i)} recovers the sample average approach when $\theta=0$ as the Wasserstein ambiguity set would collapse to a singleton set containing only the empirical distribution $\mathbb{P}^\dag$ and \emph{(ii)} recovers the classical robust optimization approach when $ \theta \geq \sup_{\bm{\zeta}_1, \bm{\zeta}_2 \in \mathcal{D}} \rho(\bm{\zeta}_1, \bm{\zeta}_2) $ as the Wasserstein ambiguity set would contain all Dirac distributions, each of which puts a unit probability mass on a possible scenario in the support set $\mathcal{D}$.

Recently, \cite{Chen_Sim_Xiong_2017} show that the Wasserstein ambiguity set $\mathcal{F}_W(\theta)$ coincides with the marginal distribution of $(\bm{p}_s, \bm{r}_s)$ under $\mathbb{P}$, for all $\mathbb{P}$ in a lifted ambiguity set $\mathcal{G}_s(\theta)$ such that 
\begin{equation*}
\label{set:Wasserstein}
\begin{array}{l}
\mathcal{G}_W(\theta) = \\
\left\{\mathbb{P} \in \mathcal{P}_0 \left(\mathbb{R}^{I_{s1}} \times \mathbb{R}^{I_{s2}} \times [N] \right) ~\left\vert~
\begin{array}{ll}
\left(\bmt{p}_s, \bmt{r}_s, \tilde{n}_s \right) \sim \mathbb{P} \\
\mathbb{E}_\mathbb{P}[\rho((\bmt{p}_s, \bmt{r}_s), (\bm{p}^\dag_{\tilde{n}_s}, \bm{r}^\dag_{\tilde{n}_s})) \mid \tilde{n}_s \in [N]] \leq \theta \\	
\pp{(\bmt{p}_s, \bmt{r}_s) \in \mathcal{D} \mid \tilde{n}_s = n} = 1 ~~\forall n \in [N] \\
\pp{\tilde{n}_s = n} = \frac{1}{N} ~~\forall n \in [N]
\end{array}
\right.\right\}
\end{array}
\end{equation*}
is in the general format~\cref{set:general format}.
\end{example}

\begin{remark}	
Statistical-distance-based ambiguity sets critically depend on the training samples by considering the discrete empirical distribution. They may lead to conservative solutions when one has a firm belief on the structure information about the uncertainty such as mean, dispersion, correlation, and dependence that commonly appear in the generalized moment-based ambiguity set. Presenting the statistical-distance-based ambiguity sets in the general format~\cref{set:general format}, we can mitigate the conservativeness by further incorporate the structure information in an intuitive way. For example, we can tailor an ambiguity set that restricts the Wasserstein distance from the ambiguous distribution to the empirical distribution, while at the same time specifies the popular mean absolute deviation from the mean (see, for example, \cite{Postek_2018}) as follows:
\begin{equation*}
\label{set:hybrid}
\begin{array}{l}
\mathcal{G}(\theta) = \\
\left\{\mathbb{P} \in \mathcal{P}_0 \left(\mathbb{R}^{I_{s1}} \times \mathbb{R}^{I_{s2}} \times [N] \right) ~\left\vert~
\begin{array}{ll}
\left(\bmt{p}_s, \bmt{r}_s, \tilde{n}_s \right) \sim \mathbb{P} \\
\ep{(\bmt{p}_s, \bmt{r}_s) \mid \tilde{n}_s \in [N]} \in [\underline{\bm{\mu}}, \bar{\bm{\mu}}] \\
\ep{\vert \bm{e}^\top((\bmt{p}_s, \bmt{r}_s) - \bm{\mu}_0) \vert \mid \tilde{n}_s \in [N]} \leq \bm{\nu} \\
\ep{\rho((\bmt{p}_s, \bmt{r}_s), (\bm{p}^\dag_{\tilde{n}_s}, \bm{r}^\dag_{\tilde{n}_s})) \mid \tilde{n}_s \in [N]} \leq \theta \\	
\pp{(\bmt{p}_s, \bmt{r}_s) \in \mathcal{D} \mid \tilde{n}_s = n} = 1 ~~\forall n \in [N] \\
\pp{\tilde{n}_s = n} = \frac{1}{N} ~~\forall n \in [N] \\
{\rm for~some~} \bm{\mu}_0 \in [\underline{\bm{\mu}}, \bar{\bm{\mu}}]
\end{array}
\right.\right\}.
\end{array}
\end{equation*}
\end{remark}

\subsection*{Mixture-distribution-based ambiguity sets}
We can use the format~\cref{set:general format} to specify an ambiguous mixture distribution, which is useful, for example, in modeling multi-modal distributions under ambiguity. Consider the ambiguity set
$$
\begin{array}{l}
\mathcal{G}_s = \\
\left\{\mathbb{P} \in \mathcal{P}_0\left(\mathbb{R}^{I_{s1}} \times \mathbb{R}^{I_{s2}} \times [N] \right) ~\left\vert~
\begin{array}{ll}
(\bmt{p}_s, \bmt{r}_s,\n_s) \sim \mathbb{P}\\
\ep{(\bmt{p}_s, \bmt{r}_s) \mid \n_s = n} = \bm{\mu}_n ~~\forall n \in [N] \\
\ep{\bm{g}_n(\bmt{p}_s, \bmt{r}_s) \mid \n_s = n} \leq \bm{\nu}_n ~~\forall n \in [N] \\
\pp{(\bmt{p}_s, \bmt{r}_s) \in \mathcal{D}_n \mid \n_s = n} = 1 ~~\forall n \in [N] \\
\pp{\n_s = n} = \omega_n ~~\forall n \in [N] \\
\mbox{for some } \bm{\omega} \in \mathcal{W}, ~(\bm{\mu}_n,\bm{\nu}_n) \in \mathcal{U}_n  ~~\forall n \in [N]
\end{array}\right.
\right\}.
\end{array}
$$ 
The projection over $(\bmt{p}_s, \bmt{r}_s)$ of any member $\mathbb{P} \in \mathcal{G}_s$ is a mixture of $N$ distinct distributions, which themselves are ambiguous. Hence, for a given $\bm{\omega} \in \mathcal{W}$, we can write
$
\Pi_{(\bmt{p}_s, \bmt{r}_s)}\mathbb{P} = \sum_{n \in [N]} \omega_n \mathbb{P}_n,
$
where each $\mathbb{P}_n \in \mathcal{F}_n$ is unknown beyond certain distributional information described in an ambiguity set
$$
\mathcal{F}_n = \left\{\mathbb{P} \in \mathcal{P}_0\left(\mathbb{R}^{I_{s1}} \times \mathbb{R}^{I_{s2}} \right) ~\left\vert~
\begin{array}{ll}
(\bmt{p}_s, \bmt{r}_s) \sim \mathbb{P}\\
\ep{(\bmt{p}_s, \bmt{r}_s)} = \bm{\mu}_n \\
\ep{\bm{g}_n(\bmt{p}_s, \bmt{r}_s)} \leq \bm{\nu}_n \\
\pp{(\bmt{p}_s, \bmt{r}_s) \in \mathcal{D}_n} = 1\\
(\bm{\mu}_n,\bm{\nu}_n) \in \mathcal{U}_n
\end{array}\right.
\right\}.
$$
It is possible to modify the ambiguity sets $\mathcal{F}_n$, $n \in [N]$ to specify first-and second- order moments of their member distributions. As a result, the ambiguity set $\mathcal{G}_s$ would include mixture Gaussian distributions with the specified moments and the corresponding distributionally robust MDPs hedge against these mixture Gaussian distributions.

\section{Finite Horizon Case}\label{section3}
In this section, we focus on distributionally robust MDPs with a finite number of decision stages, i.e., $T<\infty$. We propose a decision criterion we term distributionally robustness, which incorporates a priori information of how parameters are distributed. We show that a strategy defined through backward induction, which we called $\mathcal{S}$-robust strategy, is a distributionally robust strategy. We further show that such strategy is solvable in polynomial time under mild technical conditions.

Similarly to \cite{Nilim_Ghaoui_2005}, we assume that when a state is visited multiple times, the uncertain parameters may take a different realization for each visit. Therefore, multiple visits to a state can be treated as visiting different states. By
introducing dummy states, for finite horizon case we make the following assumption to simplify our derivations, without any loss of generality.

\begin{assumption}\label{assu_mdp}
We assume:
\begin{enumerate}
\item Each state belongs to only one stage, i.e., $\mathcal{S}=\bigcup_{t=1}^T\mathcal{S}_t$ where $\mathcal{S}_t$ is the set of states that belong to the $t$-th stage for all $ t \in [T]$.
\item The first stage only contains one state $s_1$.
\item The terminal reward equals zero.
\end{enumerate}
\end{assumption}

Under \cref{assu_mdp}, we can partition the state space $\mathcal{S}$ according to the stage to which each state belongs. We next define the notion of distributionally robustness for a strategy.

\begin{definition}\label{definition1}
A strategy $\bm{\pi}^*\in\Pi^{\text{HR}}$ is distributionally robust with respect to $\mathcal{G_S}$ if it satisfies 
$$\inf_{\mathbb{P} \in \mathcal{G_S}}w(\bm{\pi},\mathbb{P},s_1)\leq\inf_{\mathbb{P}'\in \mathcal{G_S}}w(\bm{\pi}^*,\mathbb{P}',s_1) 
~~\forall \bm{\pi} \in\Pi^{\text{HR}}.
$$
\end{definition}

In words, each strategy is evaluated by its expected performance under the (respective) most adversarial distribution of the uncertain parameters, and a distributionally robust strategy is the optimal strategy according to this measure. The main focus of this section is to derive approaches for obtaining the distributionally robust strategy for the finite horizon case. To this end, we need the following definition.

\begin{definition}\label{def}
Given an undiscounted distributionally robust MDP $\langle \mathcal{S}, \mathcal{A}, T, \gamma, \mathcal{G}_{\mathcal{S}}\rangle$ with $T<\infty$, we define the $\mathcal{S}$-robust problem through the following:
\begin{enumerate}
\item For all $s\in \mathcal{S}_T$, the $\mathcal{S}$-robust value ${v}_T(s)\triangleq 0$.
\item For all $s\in \mathcal{S}_t$, $t<T$,  the $\mathcal{S}$-robust value ${v}_t(s)$ is defined as 
\begin{equation*}\label{S-robust problem 1}
\begin{array}{ll}
&{v}_t(s) \\
\triangleq & \displaystyle \max_{\bm{\pi}_s \in \mathcal{P}(\mathcal{A}_s)} \inf_{\mathbb{P} \in \mathcal{G}_s} \ep{\sum_{a \in \mathcal{A}_s} \pi_s(a) \left(\tilde{r}(s, a) + \sum_{s' \in \mathcal{S}_{t+1}} \tilde{p}(s' \vert s, a) v_{t+1}(s')\right)} \\
=& \displaystyle \max_{\bm{\pi}_s \in \mathcal{P}(\mathcal{A}_s)} \inf_{\mathbb{P} \in \mathcal{G}_s} \ep{\bmt{r}^\top_s \bm{\pi}_s + \bmt{p}^\top_s \bm{V}_{t+1,s} \bm{\pi}_s},\\
\end{array}
\end{equation*}
and the $\mathcal{S}$-robust randomized action $\bm{\pi}^*_s$ is defined as 
\begin{equation*}\label{S-robust problem 2}
\bm{\pi}^*_s \in
\argmax_{\bm{\pi}_s \in \mathcal{P}(\mathcal{A}_s)} \inf_{\mathbb{P} \in \mathcal{G}_s} \ep{\bmt{r}^\top_s \bm{\pi}_s + \bmt{p}^\top_s \bm{V}_{t+1,s} \bm{\pi}_s},
\end{equation*}
where $\bm{V}_{t+1,s} \in \mathbb{R}^{|\mathcal{S}_{t+1}||\mathcal{A}_s| \times |\mathcal{A}_s|}$ is defined as  
$$
\bm{V}_{t+1,s} \triangleq \left[ 
\begin{array}{ccccc}
\bm{v}_{t+1} ~&\bm{0} ~&\bm{0} ~&\cdots ~&\bm{0} \\
\bm{0} ~&\bm{v}_{t+1} ~&\bm{0} ~&\cdots ~&\bm{0} \\
\vdots ~&\vdots ~&\vdots ~&\ddots ~&\vdots \\
\bm{0} ~&\bm{0} ~&\bm{0} ~&\cdots ~&\bm{v}_{t+1} \\
\end{array}
\right] = \left[
\begin{array}{c}
\bm{v}_{t+1}\bm{e}_1^\top \\
\bm{v}_{t+1}\bm{e}_2^\top \\
\vdots \\
\bm{v}_{t+1}\bm{e}_{\vert \mathcal{A}_s \vert}^\top \\
\end{array}
\right]  
$$	
with $ \bm{v}_{t+1} = \left(v_{t+1}(s')\right)_{s' \in \mathcal{S}_{t+1}}$ and $ \bm{e}_1, \dots, \bm{e}_{\vert \mathcal{A}_s \vert}$ being the standard basis in $\mathbb{R}^{\vert \mathcal{A}_s \vert}$. 

\item A strategy $\bm{\pi}^*$ is a $\mathcal{S}$-robust strategy if for any $ s\in \mathcal{S}$ and any history $ \mathcal{H}_t$ ends at $s$, we have $\bm{\pi}^*$, conditioned on history $ \mathcal{H}_t$, is a $\mathcal{S}$-robust randomized action.
\end{enumerate}
\end{definition}

Note that the definition essentially requires that the strategy must be robust with respect to each subproblem over time $t\in\{1,\dots,T\}$, and hence the name ``$\mathcal{S}$-robust''. Indeed, readers familiar with literature in robust MDPs (see \cite{Iyengar_2005, Nilim_Ghaoui_2005, White_1994}) may find that $\mathcal{S}$-robust strategy is the solution to the robust MDP where the ambiguity set only contains the support information of uncertain parameters $\mathcal{G}_s=\{\mathbb{P} \mid \pp{(\bmt{p}_s, \bmt{r}_s, \bmt{u}_s) \in \mathcal{D}_{s}} = 1\}$. The following theorem shows that any $\mathcal{S}$-robust strategy $\pi^*$ is distributionally robust.

\begin{theorem}\label{theorem:S robust}
Let $T<\infty$. Under \cref{assu_mdp}, if $\bm{\pi}^*$ is a $\mathcal{S}$-robust strategy, then
\begin{enumerate}
\item $\bm{\pi}^*$ is a distributionally robust strategy with respect to $\mathcal{G_S}$.
\item There exists $\mathbb{P}^*\in \mathcal{G}_s$ such that $(\bm{\pi}^*,\mathbb{P}^*)$ is a saddle point. That is,
\begin{equation*}
\max_{\bm{\pi}\in \Pi^{\textnormal{HR}}}w(\bm{\pi},\mathbb{P}^*,s_1)=w(\bm{\pi}^*,\mathbb{P}^*,s_1)=\inf_{\mathbb{P}\in \mathcal{G_S}}w(\bm{\pi}^*,\mathbb{P},s_1).
\end{equation*}
\end{enumerate}
\end{theorem}

\begin{proof}[Proof]
We proceed in three steps. We first show that the expected performance of a given strategy under a distribution $\mathbb{P}\in \mathcal{G_S}$ depends only on the expected value of the uncertain parameters~(Step~1). This allows us to reduce distributionally robust MDPs to classical robust MDPs which considers the set-inclusive formulation of uncertainty. We then show that the set of expected value of the uncertain parameters is convex and compact~(Step~2). Finally, we complete the proof by using results in classical robust MDPs~(Step~3).

\emph{Step~1.} We use Lemma~3.2 in \cite{Xu_Mannor_2012}, which states that given $\bm{\pi}\in \Pi^{\textnormal{HR}}$ and $\mathbb{P} \in \mathcal{G_S}$, we have $w(\bm{\pi},\mathbb{P},s_1)=u(\bm{\pi},\bmb{p},\bmb{r},s_1)$, where $\bmb{p}=\ep{\bmt{p}}$ and $\bmb{r}=\ep{\bmt{r}}$. Thus distributionally robust MDPs reduces to classical robust MDPs. 

\emph{Step~2.}  We characterize the set of expected value of the parameters. It is not hard to verify that by the design of format~\cref{set:general format}, the ambiguity set $\mathcal{G}_s$ is convex. In addition, because the feasible set of each constraint in $\mathcal{G}_s$ is weakly closed (and so is the intersection of feasible sets of these constraints), $\mathcal{G}_s$ is also weakly closed. Since $\mathcal{G}_s$ is bounded, it is compact. Hence the set ${\mathcal{Y}}_s\triangleq\{\mathbb{E}_{\mathbb{P}}[(\bmt{p}_s, \bmt{r}_s, \n_s)] \mid \mathbb{P}\in \mathcal{G}_s\}$, as the image of $\mathcal{G}_s$ under expectation (which is {a continuous function}), is closed. Following from the assumed boundedness of $\mathcal{D}_n$, $n \in [N_s]$, $\mathcal{Y}_s$ is also bounded and is thus compact. Furthermore, given any $s\in \mathcal{S}$, $\mathbb{P}_1,\mathbb{P}_2\in \mathcal{G}_s$ and $\lambda\in[0,1]$, we have $\lambda\mathbb{E}_{\mathbb{P}_1}[(\bmt{p}_s, \bmt{r}_s, \n_s)]+(1-\lambda)\mathbb{E}_{\mathbb{P}_2}[(\bmt{p}_s, \bmt{r}_s, \n_s)]=\mathbb{E}_{\lambda\mathbb{P}_1 +(1-\lambda)\mathbb{P}_2}[(\bmt{p}_s, \bmt{r}_s, \n_s)]$. Thus the set $\mathcal{Y}_s$ is convex since $\mathcal{G}_s$ is convex. Therefore, the set $\mathcal{Z}_s\triangleq\{\mathbb{E}_{\mathbb{P}}[(\bmt{p}_s, \bmt{r}_s)] \mid \mathbb{P}\in \mathcal{G}_s\}$, as the projection of $\mathcal{Y}_s$ onto its first two coordinates, is convex and compact. We are now able to conclude that for any $s\in \mathcal{S}$ and $\bm{\pi}_s\in \mathcal{P}(\mathcal{A}_s)$, there exists $(\bm{p}_s^*, \bm{r}_s^*)\in \mathcal{Z}_s$ that satisfies 
$$
\displaystyle\inf_{(\bm{p}_s, \bm{r}_s)\in\mathcal{Z}_s}\left\{\bm{r}^\top_s \bm{\pi}_s + \bm{p}^\top_s \bm{V}_{t+1,s} \bm{\pi}_s\right\}={\bm{r}_s^*}^\top \bm{\pi}_s + {\bm{p}_s^*}^\top \bm{V}_{t+1,s} \bm{\pi}_s
$$ 
because the objective of the minimization is linear in $\bm{p}_s$ and $\bm{r}_s$ and the constraint set is convex and compact. Moreover, we can construct $\bigotimes_{s\in \mathcal{S}}\mathcal{Z}_s=\{\ep{(\bmt{p}, \bmt{r})} \mid \mathbb{P}\in \mathcal{G_S}\}$ by the state-wise decomposability of the lifted ambiguity set $\mathcal{G_S}$ in~\cref{set:decomposite}. 

\emph{Step~3.} We now explore the equivalence between distributionally robust MDPs and classical robust MDPs, where the uncertainty set is $\bigotimes_{s\in \mathcal{S}}\mathcal{Z}_s$. It is well known that for classical robust MDPs, a saddle point of the minimax objective 
$$\max_{\bm{\pi}\in \Pi^{\text{HR}}}\inf_{(\mathbf{p}, \mathbf{r})\in\bigotimes_{s\in \mathcal{S}}\mathcal{Z}_s}u(\bm{\pi},\bm{p},\bm{r},s_1)$$
exists (see \cite{Iyengar_2005} and \cite{Nilim_Ghaoui_2005}). To be more precise, there exist 
$\bm{\pi}^*\in\Pi^{\text{HR}}$ and $(\bm{p}^*, \bm{r}^*)\in \bigotimes_{s\in \mathcal{S}}\mathcal{Z}_s$ satisfying
\begin{equation}
\nonumber
\displaystyle\inf_{(\mathbf{p}, \mathbf{r})\in\bigotimes_{s\in \mathcal{S}}\mathcal{Z}_s}u(\bm{\pi}^*,\bm{p},\bm{r},s_1) =u(\bm{\pi}^*,\bm{p}^*,\bm{r}^*,s_1)=\max_{\bm{\pi}\in \Pi^{\text{HR}}}u(\bm{\pi},\bm{p}^*,\bm{r}^*,s_1).
\end{equation}
Moreover, we can construct $\bm{\pi}^*$ and $(\bm{p}^*,\bm{r}^*)$ state-wise through defining $\bm{\pi}^*=\bigotimes_{s\in \mathcal{S}}\bm{\pi}^*_s$ and $(\bm{p}^*,\bm{r}^*)=\bigotimes_{s\in \mathcal{S}}(\bm{p}^*_s,\bm{r}^*_s)$, where for each $s\in \mathcal{S}_t$, $\bm{\pi}^*_s$ and $(\bm{p}^*_s,\bm{r}^*_s)$ solves the following zero-sum game
\begin{equation*}
\max_{\bm{\pi}_s \in \mathcal{P}(\mathcal{A}_s)} \inf_{(\bm{p}_s, \bm{r}_s)\in \mathcal{Z}_s} \left\{\bm{r}^\top_s \bm{\pi}_s + \bm{p}^\top_s \bm{V}_{t+1,s} \bm{\pi}_s\right\}.
\end{equation*}
Thus $\bm{\pi}_s$ is any $\mathcal{S}$-robust randomized action, and hence $\bm{\pi}$ can be any $\mathcal{S}$-robust strategy. In Step~2, we know that there exists $\mathbb{P}_s^*\in \mathcal{G}_s$ satisfying $\mathbb{E}_{\mathbb{P}_s^*}[(\bmt{p}_s, \bmt{r}_s)]=(\bm{p}^*_s, \bm{r}^*_s)$. Letting $\mathbb{P}^*=\bigotimes_{s\in \mathcal{S}}\mathbb{P}^*_s$ and using Lemma~3.2 in~\cite{Xu_Mannor_2012}, we obtain
\begin{equation*}
\max_{\bm{\pi}\in\Pi^{\text{HR}}}w(\bm{\pi},\mathbb{P}^*,s_1)=\max_{\bm{\pi}\in \Pi^{\text{HR}}}u(\bm{\pi},\bm{p}^*,\bm{r}^*,s_1),
~w(\bm{\pi}^*,\mathbb{P}^*,s_1)=u(\bm{\pi}^*,\bm{p}^*,\bm{r}^*,s_1)
\end{equation*}
and 
\begin{equation*}
\inf_{\mathbb{P}\in \mathcal{G_S}}w(\bm{\pi},\mathbb{P},s_1)=\displaystyle\inf_{(\mathbf{p}, \mathbf{r})\in\bigotimes_{s\in \mathcal{S}}\mathcal{Z}_s}u(\bm{\pi}^*,\bm{p},\bm{r},s_1).
\end{equation*}
As a result,
\begin{equation*}
\max_{\bm{\pi}\in \Pi^{\text{HR}}}w(\bm{\pi},\mathbb{P}^*,s_1)=w(\bm{\pi}^*,\mathbb{P}^*,s_1) =\inf_{\mathbb{P}\in \mathcal{G_S}}w(\bm{\pi}^*,\mathbb{P},s_1).
\end{equation*}
Therefore, the second part of \cref{theorem:S robust} holds, and the first part also follows immediately.
\end{proof}

Leveraging \cref{theorem:S robust}, we now investigate the computational aspect of the $\mathcal{S}$-robust randomized action, which can be obtained from solving a sequence of classical robust optimization problems.

\begin{theorem}
\label{theorem:dro reformulation}
For all $ s \in \mathcal{S}_t $, $ t < T $, the $\mathcal{S}$-robust randomized action is the optimal solution to the following optimization problem 
\begin{equation}
\label{prob:S-robust}
\max_{\bm{\pi}_s \in \mathcal{P}(\mathcal{A}_s)} \inf_{\mathbb{P} \in \mathcal{G}_s} \ep{\bmt{r}^\top_s \bm{\pi}_s + \bmt{p}^\top_s \bm{V}_{t+1,s} \bm{\pi}_s},
\end{equation}
or equivalently, to the following classical robust optimization problem
\begin{equation}\label{prob:s-robust}
\begin{array}{rcll}
&\max & -\delta \\
&{\rm s.t.} & \delta \geq \bm{\alpha}^\top\bm{\omega} + \displaystyle \sum_{j \in [J_s]} (\bm{\beta}^\top_j\bm{\mu}_j + \bm{\gamma}_j^\top\bm{\nu}_j) ~~\forall (\bm{\mu}, \bm{\nu}, \bm{\omega}) \in \mathcal{V} \\
&& 
\begin{array}{ll}
\alpha_n + & \displaystyle \sum_{j \in \mathcal{J}_n}(\bm{\beta}_j^\top(\bm{p}_s, \bm{r}_s) + \bm{\gamma}_j^\top\bm{g}_{jn}(\bm{p}_s, \bm{r}_s)) \\
&\geq -\bm{r}^\top_s \bm{\pi}_s - \bm{p}^\top_s \bm{V}_{t+1,s} \bm{\pi}_s ~~\forall (\bm{p}_s, \bm{r}_s) \in \mathcal{D}_n, ~n \in [N_s] 
\end{array}\\
&& \bm{\pi}_s \in \mathcal{P}(\mathcal{A}_s), ~\delta \in \mathbb{R}, ~\bm{\alpha} \in \mathbb{R}^{N_s}, \bm{\beta}_j \in \mathbb{R}^{I_{s1} + I_{s2}}, ~\bm{\gamma}_j \in \mathbb{R}^{M_j}_+ ~~\forall j \in [J_s],
\end{array}
\end{equation}
where $\bm{\mu} \triangleq (\bm{\mu}_j)_{j \in [J_s]}$, $\bm{v} \triangleq (\bm{\nu}_j)_{j \in [J_s]}$ and
$$
\mathcal{V} \triangleq \left\{(\bm{\mu}, \bm{\nu}, \bm{\omega}) ~\left|~ \bm{\omega} \in \mathcal{W}, ~\dfrac{(\bm{\mu}_j, \bm{\nu}_j)}{\sum_{i \in \mathcal{N}_j}\omega_i} \in \mathcal{U}_j ~\forall j \in [J_s] \right.\right\}.
$$
\end{theorem}

\begin{proof}[Proof]
We rewrite problem~\cref{prob:S-robust} for determining the $\mathcal{S}$-robust randomized action into
$$
\begin{array}{cll}
\max & \eta \\
{\rm s.t.} & \displaystyle \inf_{\mathbb{P} \in \mathcal{G}_s} \ep{\bmt{r}^\top_s \bm{\pi}_s + \bmt{p}^\top_s \bm{V}_{t+1,s} \bm{\pi}_s} \geq \eta \\
& \eta \in \mathbb{R}, ~\bm{\pi}_s \in \mathcal{P}(\mathcal{A}_s).
\end{array}
$$
Changing variables from $ \eta $ to $ -\eta $, we obtain
\begin{equation}
\label{prob:dro}
\begin{array}{rcll}
&\max & -\eta \\
&{\rm s.t.} & \displaystyle\sup_{\mathbb{P} \in \mathcal{G}_s} \ep{-\bmt{r}^\top_s \bm{\pi}_s - \bmt{p}^\top_s \bm{V}_{t+1,s} \bm{\pi}_s} \leq \eta \\
&& \eta \in \mathbb{R}, ~\bm{\pi}_s \in \mathcal{P}(\mathcal{A}_s).
\end{array}
\end{equation}
Let $\mathcal{U} \triangleq \{(\bm{\mu}, \bm{\nu}) \mid (\bm{\mu}_j, \bm{\nu}_j) \in \mathcal{U}_j ~\forall j \in [J_s]\}$. The worst-case expectation on the left-hand side of the first constraint in~\cref{prob:dro} can be re-expressed as follows:
$$
\lambda^* = \sup_{(\bm{\mu}, \bm{\nu}, \bm{\omega}) \in \mathcal{U} \times \mathcal{W}} \lambda(\bm{\mu}, \bm{\nu}, \bm{\omega}),  
$$
where given $(\bm{\mu}, \bm{\nu}, \bm{\omega}) \in \mathcal{U} \times \mathcal{W}$, we define an ambiguity set
\begin{equation*}
\begin{array}{l}
\mathcal{G}_s(\bm{\mu}, \bm{\nu}, \bm{\omega}) \triangleq \\
\left\{\mathbb{P} \in \mathcal{P}_0(\mathbb{R}^{I_{s1}} \times \mathbb{R}^{I_{s2}} \times [N_s]) ~\left\vert~
\begin{array}{ll}
(\bmt{p}_s, \bmt{r}_s, \tilde{n}_s) \sim \mathbb{P} \\
\ep{(\bmt{p}_s, \bmt{r}_s) \mid \tilde{n}_s \in \mathcal{N}_j} = \bm{\mu}_j &~\forall j \in [J_s] \\
\ep{\bm{g}_{j\tilde{n}_s}(\bmt{p}_s, \bmt{r}_s) \mid \tilde{n}_s \in \mathcal{N}_j} \leq \bm{\nu}_j &~\forall j \in [J_s] \\
\pp{(\bmt{p}_s, \bmt{r}_s) \in \mathcal{D}_{n} \mid \tilde{n}_s = n} = 1 &~\forall n \in [N_s] \\
\pp{\tilde{n}_s = n} = \omega_n &~\forall n \in [N_s] 
\end{array}
\right.\right\}
\end{array}
\end{equation*}
and correspondingly the worst-case expectation
$$
\lambda(\bm{\mu}, \bm{\nu}, \bm{\omega}) \triangleq \sup_{\mathbb{P} \in \mathcal{G}_s(\bm{\mu}, \bm{\nu}, \bm{\omega})} \ep{-\bmt{r}^\top_s \bm{\pi}_s - \bmt{p}^\top_s \bm{V}_{t+1,s} \bm{\pi}_s}.
$$

Using the law of total probability, we can represent the joint distribution $\mathbb{P}$ of $ (\bmt{p}_s, \bmt{r}_s, \tilde{n}_s) $ as the marginal distribution $\mathbb{P}^\dag $ of $ \n_s $ supported on $[N_s]$ and the conditional distributions $ \mathbb{P}_n$ of $(\bmt{p}_s, \bmt{r}_s)$ given $ \n_s = n$, $n \in [N_s]$. Let $\bar{\omega}_j =\sum_{i \in \mathcal{N}_j} \omega_i$ for all $j \in [J_s]$, we can now reformulate $\lambda(\bm{\mu}, \bm{\nu}, \bm{\omega})$ as:
\begin{equation}
\nonumber
\begin{array}{rcll}
\lambda(\bm{\mu}, \bm{\nu}, \bm{\omega}) = &\sup &\displaystyle\sum_{n \in [N_s]} \omega_n \epn{-\bmt{r}^\top_s \bm{\pi}_s - \bmt{p}^\top_s \bm{V}_{t+1,s} \bm{\pi}_s} \\
&{\rm s.t.}& \displaystyle\sum_{n \in \mathcal{N}_j} \omega_n \epn{(\bmt{p}_s, \bmt{r}_s)} = \bar{\omega}_j\bm{\mu}_j &~\forall j \in [J_s] \\
&& \displaystyle\sum_{n \in \mathcal{N}_j} \omega_n \epn{\bm{g}_{jn}(\bmt{p}_s, \bmt{r}_s)} \leq \bar{\omega}_j\bm{\nu}_j &~\forall j \in [J_s] \\
&& \ppn{(\bmt{p}_s, \bmt{r}_s) \in \mathcal{D}_n} = 1 &~\forall n \in [N_s].
\end{array}
\end{equation}
The dual of $ \lambda(\bm{\mu}, \bm{\nu}, \bm{\omega}) $ is denoted as $\lambda_{\rm D}(\bm{\mu}, \bm{\nu}, \bm{\omega})$, which is given by
\begin{equation}
\nonumber
\begin{array}{ll}
&\left\{
\begin{array}{cll}
\inf & \displaystyle \sum_{n \in [N_s]} \alpha_n + \displaystyle \sum_{j \in [J_s]}\bar{\omega}_j(\bm{\beta}^\top_j\bm{\mu}_j + \bm{\gamma}_j^\top\bm{\nu}_j) \\
{\rm s.t.}& 
\begin{array}{ll}
\alpha_n + &\omega_n \displaystyle \sum_{j \in \mathcal{J}_n} (\bm{\beta}_j^\top(\bm{p}_s, \bm{r}_s) + \bm{\gamma}_j^\top\bm{g}_{jn}(\bm{p}_s, \bm{r}_s)) \\
&\geq \omega_n (-\bm{r}^\top_s \bm{\pi}_s - \bm{p}^\top_s \bm{V}_{t+1,s} \bm{\pi}_s) ~~\forall (\bm{p}_s, \bm{r}_s) \in \mathcal{D}_n, ~n \in [N_s]
\end{array} \\
& \bm{\alpha} \in \mathbb{R}^{N_s}, ~\bm{\beta}_j \in \mathbb{R}^{I_{s1} + I_{s2}}, ~\bm{\gamma}_j \in \mathbb{R}^{M_j}_+ ~~\forall j \in [J_s]\\[3mm]
\end{array}
\right. \\
= &\left\{
\begin{array}{cll}
\inf & \bm{\alpha}^\top\bm{\omega} + \displaystyle \sum_{j \in [J_s]}\bar{\omega}_j(\bm{\beta}_j^\top\bm{\mu}_j + \bm{\gamma}_j^\top\bm{\nu}_j) \\
{\rm s.t.}& 
\begin{array}{ll}
\alpha_n + & \displaystyle \sum_{j \in \mathcal{J}_n} (\bm{\beta}_j^\top(\bm{p}_s, \bm{r}_s) + \bm{\gamma}_j^\top\bm{g}_{jn}(\bm{p}_s, \bm{r}_s)) \\
&\geq -\bm{r}^\top_s \bm{\pi}_s - \bm{p}^\top_s \bm{V}_{t+1,s} \bm{\pi}_s ~~\forall (\bm{p}_s, \bm{r}_s) \in \mathcal{D}_n, ~n \in [N_s] 
\end{array}\\
& \bm{\alpha} \in \mathbb{R}^{N_s}, ~\bm{\beta}_j \in \mathbb{R}^{I_{s1} + I_{s2}}, ~\bm{\gamma}_j \in \mathbb{R}^{M_j}_+ ~~\forall j \in [J_s],
\end{array} 
\right.
\end{array}
\end{equation}
where the equality follows from for all $n \in [N_s]$, first changing variable from $\alpha_n$ to $\omega_n \alpha_n$ and then dividing both sides of the constraint by $\omega_n$, which is allowed since $\mathcal{W} \subseteq {\rm int}\{\bm{\omega} \in \mathbb{R}^{N_s}_+ \mid \bm{e}'\bm{\omega} = 1\}$. By weak duality, $\lambda(\bm{\mu}, \bm{\nu}, \bm{\omega}) \leq \lambda_{\rm D}(\bm{\mu}, \bm{\nu}, \bm{\omega})$. By the general min-max theorem (see \cite{Sion_1958}), we have
$$
\lambda_{\rm D}^* = \displaystyle\sup_{(\bm{\mu}, \bm{\nu}, \bm{\omega}) \in \mathcal{U} \times \mathcal{W}} \lambda_{\rm D}(\bm{\mu}, \bm{\nu}, \bm{\omega}) \leq \lambda_{\rm P}^*,
$$
where the min-max problem $\lambda_{\rm P}^*$ corresponds to the max-min problem $\lambda_{\rm D}^*$ is
\begin{equation}
\label{prob:lambda dual}
\lambda_{\rm P}^* 
\triangleq 
\left\{
\begin{array}{cll}
\inf & \delta \\
{\rm s.t.} & \delta \geq \bm{\alpha}^\top\bm{\omega} + \displaystyle \sum_{j \in [J_s]} \bar{\omega}_j(\bm{\beta}^\top_j\bm{\mu}_j + \bm{\gamma}_j^\top\bm{\nu}_j) ~~\forall \bm{\omega} \in \mathcal{W}, (\bm{\mu}_j, \bm{\nu}_j) \in \mathcal{U}_j, ~j \in [J_s]\\ 
& 
\begin{array}{ll}
\alpha_n + &\displaystyle \sum_{j \in \mathcal{J}_n}(\bm{\beta}_j^\top(\bm{p}_s, \bm{r}_s) + \bm{\gamma}_j^\top\bm{g}_{jn}(\bm{p}_s, \bm{r}_s)) \\
&\geq -\bm{r}^\top_s \bm{\pi}_s - \bm{p}^\top_s \bm{V}_{t+1,s} \bm{\pi}_s ~~\forall (\bm{p}_s, \bm{r}_s) \in \mathcal{D}_n, ~n \in [N_s] 
\end{array}\\
& \delta \in \mathbb{R}, ~\bm{\alpha} \in \mathbb{R}^{N_s}, ~\bm{\beta}_j \in \mathbb{R}^{I_{s1} + I_{s2}}, ~\bm{\gamma}_j \in \mathbb{R}^{M_j}_+ ~~\forall j \in [J_s].
\end{array}
\right.
\end{equation}
Due to the presence of products of uncertain variables (e.g., $\bar{\omega}_j\bm{\beta}_j^\top\bm{\mu}_j$), problem~\cref{prob:lambda dual} is nonconvex. Since $\bm{\omega}>0$ (and hence $\bar{\omega}_j >0$), an equivalent convex representation can be obtained by changing variables in problem~\cref{prob:lambda dual} from $ \bar{\omega}_j(\bm{\mu}_j, \bm{\nu}_j) $ to $(\bm{\mu}_j, \bm{\nu}_j)$ for all $ j \in [J_s] $:
\begin{equation*}
\lambda_{\rm P}^* = 
\left\{
\begin{array}{rcll}
\inf & \delta \\
&{\rm s.t.} & \delta \geq \bm{\alpha}^\top\bm{\omega} + \displaystyle \sum_{j \in [J_s]} (\bm{\beta}^\top_j\bm{\mu}_j + \bm{\gamma}_j^\top\bm{\nu}_j) ~~\forall (\bm{\mu}, \bm{\nu}, \bm{\omega}) \in \mathcal{V} \\ 
&& 
\begin{array}{ll}
\alpha_n + &\displaystyle \sum_{j \in \mathcal{J}_n}(\bm{\beta}_j^\top(\bm{p}_s, \bm{r}_s) + \bm{\gamma}_j^\top\bm{g}_{jn}(\bm{p}_s, \bm{r}_s)) \\
&\geq -\bm{r}^\top_s \bm{\pi}_s - \bm{p}^\top_s \bm{V}_{t+1,s} \bm{\pi}_s ~~\forall (\bm{p}_s, \bm{r}_s) \in \mathcal{D}_n, ~n \in [N_s] 
\end{array}\\
&& \delta \in \mathbb{R}, ~\bm{\alpha} \in \mathbb{R}^{N_s}, ~\bm{\beta}_j \in \mathbb{R}^{I_{s1} + I_{s2}}, ~\bm{\gamma}_j \in \mathbb{R}^{M_j}_+ ~~\forall j \in [J_s].
\end{array}
\right.
\end{equation*}

In fact, it can be shown that under \cref{assu_slater}, $ \lambda^* = \lambda^*_{\rm D} = \lambda^*_{\rm P} $ (see Theorem~2 in \cite{Chen_Sim_Xiong_2017}). Therefore, by re-inserting $\lambda^*_{\rm P}$ (that is, $\lambda^*$) into~\cref{prob:dro} and eliminating the decision variable $\eta$, we the desired reformulation~\cref{prob:s-robust} to solve the $ \mathcal{S}$-robust randomized action.
\end{proof}

\section{Infinite Horizon Case}\label{section5}
In this section, we study distributionally robust MDPs in the discounted-reward infinite horizon setup. In particular, we 
generalize the notion of $\mathcal{S}$-robust strategy proposed in  \cref{section3} to infinite-horizon case and show that the $\mathcal{S}$-robust strategy is distributionally robust.

Unlike the finite horizon case, we cannot model the system as (i) having finitely many states, and (ii) each visited at most once. In contrast, we have to relax either one of these two assumptions, which lead to two different natural formulations. Similarly to \cite{Xu_Mannor_2012} and \cite{Yu_Xu_2016}, we consider two models, namely, the non-stationary model and the stationary model. These two formulations can model different setups: if the system, more specifically the probability distribution of uncertain parameters, evolves with time, then the non-stationary model is more appropriate; while if the system is static, then a stationary model is preferable. For an in-depth discussion of the distinction between non-stationary and stationary models, we refer the readers to \cite{Nilim_Ghaoui_2005}.

The $\mathcal{S}$-robust strategy for the infinite horizon distributionally robust MDP is defined as follows.
\begin{definition}\label{definition4}
Given a distributionally robust MDP $\langle \mathcal{S}, \mathcal{A}, T, \gamma, \mathcal{G}_{\mathcal{S}}\rangle$ with $T=\infty$, we define the $\mathcal{S}$-robust problem through the following:
\begin{enumerate}
\item For all $s\in\mathcal{S}$, the $\mathcal{S}$-robust value ${v}_{\infty}(s)$ and $\mathcal{S}$-robust randomized action $\bm{\pi}^*_s$ are defined as
$${v}_{\infty}(s)\triangleq \displaystyle \max_{\bm{\pi}_s \in \mathcal{P}(\mathcal{A}_s)} \inf_{\mathbb{P} \in \mathcal{G}_s} \ep{\sum_{a \in \mathcal{A}_s} \pi_s(a) \left(\tilde{r}(s, a) + \gamma\sum_{s' \in \mathcal{S}} \tilde{p}(s' \vert s, a) v_{\infty}(s')\right)},$$
$$
{\bm{\pi}}^*_s\in\argmax_{\bm{\pi}_s\in\mathcal{P}(\mathcal{A}_s)}\inf_{\mathbb{P}\in \mathcal{G}_s}\ep{\sum_{a \in \mathcal{A}_s} \pi_s(a) \left(\tilde{r}(s, a) + \gamma\sum_{s' \in \mathcal{S}} \tilde{p}(s' \vert s, a) v_{\infty}(s')\right)}.
$$
					
\item A strategy ${\bm{\pi}}^*$ is a $\mathcal{S}$-robust strategy if ${\bm{\pi}}_s^*$ is a $\mathcal{S}$-robust randomized action for all $ s\in \mathcal{S}$.
\end{enumerate}
\end{definition}

To see that the $\mathcal{S}$-robust strategy is well defined, it suffices to show that the following operator $\mathcal{L}:\mathbb{R}^{|\mathcal{S}|}\rightarrow\mathbb{R}^{|\mathcal{S}|}$ is a $\gamma$-contraction with respect to $\|\cdot\|_{\infty}$ norm, i.e.,  $\|\mathcal{L}\bm{v}_1-\mathcal{L}\bm{v}_2\|_\infty\leq\gamma\|\bm{v}_1-\bm{v}_2\|_\infty$ for any $\bm{v}_1,\bm{v}_2\in\mathbb{R}^{|\mathcal{S}|}$. We define
$$\{\mathcal{L}\bm{v}\}(s)\triangleq\max_{\bm{\pi}_s\in\mathcal{P}(\mathcal{A}_s)}\inf_{\mathbb{P}\in \mathcal{G}_s}\{\mathcal{L}^{\bm{\pi}_s}_{\mathbb{P}}\bm{v}\}(s),$$ 
where for fixed $\bm{\pi}_s\in\mathcal{P}(\mathcal{A}_s)$ and $\mathbb{P}\in \mathcal{G}_s$,
\begin{equation*}
\begin{aligned}
\{\mathcal{L}_{\mathbb{P}}^{\bm{\pi}_s}\bm{v}\}(s)&\triangleq\ep{\sum_{a \in \mathcal{A}_s} \pi_s(a) \left(\tilde{r}(s, a) + \gamma\sum_{s' \in \mathcal{S}} \tilde{p}(s' \vert s, a) v(s')\right)}\\
&=\sum_{a\in \mathcal{A}_s}\pi_s(a)r(s,a)+\gamma\sum_{a\in \mathcal{A}_s}\sum_{s'\in \mathcal{S}}\pi_s(a)p(s'|s,a)v(s').
\end{aligned}
\end{equation*}

\begin{lemma}\label{lem3}
Under \cref{assu_mdp}, $\mathcal{L}$ is a $\gamma$-contraction with respect to $\|\cdot\|_{\infty}$ norm.
\end{lemma}

\begin{proof}[Proof]
Given any $\bm{v}_1,\bm{v}_2\in\mathbb{R}^{|\mathcal{S}|}$ and any $s\in \mathcal{S}$, let $(\bm{\pi}_s^1,\mathbb{P}_1)$ and $(\bm{\pi}_s^2,\mathbb{P}_2)$ be the respective saddle points for $\{\mathcal{L}\bm{v}_1\}(s)$ and $\{\mathcal{L}\bm{v}_2\}(s)$. Suppose that $\{\mathcal{L}\bm{v}_1\}(s)\geq\{\mathcal{L}\bm{v}_2\}(s)$, we have
\begin{equation*}
	\begin{array}{l l}
		0&\leq\{\mathcal{L}\bm{v}_1\}(s)-\{\mathcal{L}\bm{v}_2\}(s)\\&=\{\mathcal{L}^{\bm{\pi}^1_s}_{\mathbb{P}_1}\bm{v}_1\}(s)-\{\mathcal{L}^{\bm{\pi}^2_s}_{\mathbb{P}_2}\bm{v}_2\}(s)\\
		&\leq\{\mathcal{L}^{\bm{\pi}^1_s}_{\mathbb{P}_2}\bm{v}_1\}(s)-\{\mathcal{L}^{\bm{\pi}^1_s}_{\mathbb{P}_2}\bm{v}_2\}(s)\\
		&=\gamma\displaystyle\sum_{a\in \mathcal{A}_s}\sum_{s'\in \mathcal{S}}\pi_s^1(a)p(s'|s,a)(v_1(s')-v_2(s'))\\
		&\leq\gamma\displaystyle\sum_{a\in \mathcal{A}_s}\sum_{s'\in \mathcal{S}}\pi_s^1(a)p(s'|s,a)\|\bm{v}_1-\bm{v}_2\|_\infty\\
		&=\gamma\|\bm{v}_1-\bm{v}_2\|_\infty;
	\end{array}
\end{equation*}
here the last equality follows from the fact that $\bm{\pi}_s^1$ is a vector on the probability simplex $ \mathcal{P}(\mathcal{A}_s)$ and $\sum_{s'\in \mathcal{S}}p(s'|s,a)=1$. Similarly, for $\{\mathcal{L}\bm{v}_1\}(s)\leq\{\mathcal{L}\bm{v}_2\}(s)$, we arrive at 
$
|\{\mathcal{L}\bm{v}_1\}(s)-\{\mathcal{L}\bm{v}_2\}(s)|\leq\gamma\|\bm{v}_1-\bm{v}_2\|_\infty 
$
for all $s\in \mathcal{S}$. Taking the supremum over $s$ yields the desired result.
\end{proof}
Note that given $\bm{v}$ and $s$, by applying \cref{theorem:dro reformulation}, the $\mathcal{S}$-robust strategy can be obtained efficiently. Banach Fixed-Point Theorem states that, there exists a unique $\bm{v}_{\infty}^*$ such that $\mathcal{L}\bm{v}^*_{\infty}=\bm{v}^*_{\infty}$, which is the $\mathcal{S}$-robust value by definition. Moreover, for an arbitrary $\bm{v}^0$, the value vector sequence $\{v^n\}$ defined by $v^{n+1}=\mathcal{L}v^n=\mathcal{L}^{n+1}\bm{v}^0$ converges to $\bm{v}_{\infty}^*$ at exponential rate (see Theorem 6.2.3. in \cite{Puterman_2014}). Therefore, as the following lemma shows, we can compute the $\mathcal{S}$-robust randomized action for each $s$ (and hence $\mathcal{S}$-robust strategy) using this procedure.

\begin{lemma}\label{lem4}
For $s\in \mathcal{S}$, let $\bm{v}^n=\mathcal{L}^n\bm{v}^0$ for all $n \geq 0$ and 
$$\bm{\pi}^n_s\in\argmax_{\bm{\pi}_s\in\mathcal{P}(\mathcal{A}_s)}\inf_{\mathbb{P}\in \mathcal{G}_s}\ep{\sum_{a \in \mathcal{A}_s} \pi_s(a) \left(\tilde{r}(s, a) + \gamma\sum_{s' \in \mathcal{S}} \tilde{p}(s' \vert s, a) v^n(s')\right)}.$$
Then the sequence $\{\bm{\pi}_s^n\}_{n=1}^{\infty}$ has a convergent subsequence, and any of its limit points is a $\mathcal{S}$-robust randomized action of state $s$.
\end{lemma}

\begin{proof}[Proof]
Since $\mathcal{P}(\mathcal{A}_s)$ is compact, the sequence $\{\bm{\pi}_s^n\}_{n=1}^{\infty}$ has a convergent subsequence. To show that any limiting point is a $\mathcal{S}$-robust randomized action, we first assume $\bm{\pi}_s^n\rightarrow\bm{\pi}_s^*$ without loss of generality. We note that given any $ \mathbb{P} \in \mathcal{G}_s $ and $ \hat{\bm{\pi}}_s \in \mathcal{P}(\mathcal{A}_s) $, the operator $\mathcal{L}_{\mathbb{P}}^{\hat{\bm{\pi}}_s}$ is a $\gamma$-contraction (see \cite{Puterman_2014}).
Thus by the definition of the maximizer $\bm{\pi}^n_s$, for any $\hat{\bm{\pi}}_s\in \mathcal{P}(\mathcal{A}_s)$ and $ \bm{v}^n \in \mathbb{R}^{\vert \mathcal{S} \vert} $, we have 
$$\inf_{\mathbb{P}\in \mathcal{G}_s}\{\mathcal{L}_{\mathbb{P}}^{\hat{\bm{\pi}}_s}\bm{v}^n\}(s)\leq \inf_{\mathbb{P}'\in \mathcal{G}_s}\{\mathcal{L}_{\mathbb{P}'}^{{\bm{\pi}}^n_s}\bm{v}^n\}(s).$$
Moreover, for $\bm{v}^*_{\infty}$ defined by $\mathcal{L}\bm{v}^*_{\infty}=\bm{v}_{\infty}^*$, we have
\begin{equation}\label{equ1}
\begin{array}{ll}
~~\displaystyle\inf_{\mathbb{P}\in \mathcal{G}_s}\{\mathcal{L}_{\mathbb{P}}^{\bm{\pi}^n_s}\bm{v}^*_{\infty}\}(s)-\inf_{\mathbb{P}\in \mathcal{G}_s}\{\mathcal{L}_{\mathbb{P}}^{\hat{\bm{\pi}}_s}\bm{v}^*_{\infty}\}(s)\\
=\Big(\displaystyle\inf_{\mathbb{P}\in \mathcal{G}_s}\{\mathcal{L}_{\mathbb{P}}^{\bm{\pi}^n_s}\bm{v}^*_{\infty}\}(s)-\displaystyle\inf_{\mathbb{P}\in \mathcal{G}_s}\{\mathcal{L}_{\mathbb{P}}^{\bm{\pi}^n_s}\bm{v}^n\}(s)\Big)
+\Big(\displaystyle\inf_{\mathbb{P}\in \mathcal{G}_s}\{\mathcal{L}_{\mathbb{P}}^{\bm{\pi}^n_s}\bm{v}^n\}(s)-\displaystyle\inf_{\mathbb{P}\in \mathcal{G}_s}\{\mathcal{L}_{\mathbb{P}}^{\hat{\bm{\pi}}_s}\bm{v}^*_{\infty}\}(s)\Big)\\
\geq\Big(\displaystyle\inf_{\mathbb{P}\in \mathcal{G}_s}\{\mathcal{L}_{\mathbb{P}}^{\bm{\pi}^n_s}\bm{v}^*_{\infty}\}(s)-\displaystyle\inf_{\mathbb{P}\in \mathcal{G}_s}\{\mathcal{L}_{\mathbb{P}}^{\bm{\pi}^n_s}\bm{v}^n\}(s)\Big)
+\Big(\displaystyle\inf_{\mathbb{P}\in \mathcal{G}_s}\{\mathcal{L}_{\mathbb{P}}^{\bm{\pi}^n_s}\bm{v}^n\}(s)-\displaystyle\inf_{\mathbb{P}\in \mathcal{G}_s}\{\mathcal{L}_{\mathbb{P}}^{\bm{\pi}^n_s}\bm{v}^*_{\infty}\}(s)\Big)\\
\geq-2\gamma\|\bm{v}^n-\bm{v}^*_{\infty}\|_{\infty}.
\end{array}
\end{equation}
Let $\mathbb{P}^*=\arginf_{\mathbb{P}\in \mathcal{G}_s}\mathcal{L}_{\mathbb{P}}^{\hat{\bm{\pi}}_s}\bm{v}^*_{\infty}(s)$. By  Step 2 in the proof of \cref{theorem:S robust}, we have
\begin{equation}\label{equ2}
\begin{aligned}
\lim\limits_{n\rightarrow\infty}\inf_{\mathbb{P}\in \mathcal{G}_s}\{\mathcal{L}_{\mathbb{P}}^{\bm{\pi}^n_s}\bm{v}^*_{\infty}\}(s) \leq \lim\limits_{n\rightarrow\infty}\{\mathcal{L}_{\mathbb{P}^*}^{\bm{\pi}^n_s}\bm{v}^*_{\infty}\}(s) = \{\mathcal{L}_{\mathbb{P}^*}^{\bm{\pi}^*_s}\bm{v}^*_{\infty}\}(s) = \inf_{\mathbb{P}'\in \mathcal{G}_s}\{\mathcal{L}_{\mathbb{P}'}^{\bm{\pi}^*_s}\bm{v}^*_{\infty}\}(s);
\end{aligned}
\end{equation}
here the first equality holds since $\{\mathcal{L}_{\mathbb{P}}^{\bm{\pi}_s}\bm{v}\}(s)$ is continuous on $\bm{\pi}_s$ and $\bm{\pi}_s^n\rightarrow\bm{\pi}_s^*$, and the second equality holds due to the definition of $\mathbb{P}^*$. Combining~\cref{equ1} and~\cref{equ2}, and noting that $\bm{v}^n\rightarrow\bm{v}^*_{\infty}$ due to \cref{lem3}, we arrive at
$$\inf_{\mathbb{P}'\in \mathcal{G}_s}\{\mathcal{L}_{\mathbb{P}'}^{\bm{\pi}^*_s}\bm{v}^*_{\infty}\}(s)\geq\inf_{\mathbb{P}\in \mathcal{G}_s}\{\mathcal{L}_{\mathbb{P}}^{\hat{\bm{\pi}}_s}\bm{v}^*_{\infty}\}(s),$$
which concludes that $\bm{\pi}^*_s$ is a $\mathcal{S}$-robust randomized action of state $s$.
\end{proof}

\subsection{Non-Stationary Model}
The non-stationary model treats the system as having infinitely many states, each visited at most once. Therefore, we consider an equivalent MDP with an augmented state space, where each augmented state is defined by a pair $(s,t)$ where $s\in \mathcal{S}$ and $t\in\{1,2,\dots\}$, meaning state $s$ in the $ t $-th horizon. We define the non-stationary ambiguity set as the Cartesian product of the admissible distributions of each (augmented) state. That is,
$$\bar{\mathcal{G}}^{\infty}_\mathcal{S}\triangleq \left\{\mathbb{P}~\left\vert~\mathbb{P}=\bigotimes_{s\in \mathcal{S},t=1,2,\dots}\mathbb{P}_{s,t}, ~\mathbb{P}_{s,t}\in \mathcal{G}_s ~\forall s\in \mathcal{S},~t=1,2,\dots \right.\right\}.$$
This model is attractive as one can solve the corresponding MDP using the robust dynamic programming algorithm (see \cite{Iyengar_2005} and \cite{Nilim_Ghaoui_2005}).

The following theorems show that the $\mathcal{S}$-robust strategy is distributionally robust for the non-stationary  model. Indeed, the result is similar to Theorem~4.1 of \cite{Xu_Mannor_2012}. 
\begin{theorem}\label{thm3}
Under \cref{assu_mdp}, any $\mathcal{S}$-robust strategy is distributionally robust with respect to $\bar{\mathcal{G}}_{\mathcal{S}}^{\infty}$.
\end{theorem}

\begin{proof}[Proof]
First, we introduce the $\widehat{T}$-truncated problem below:
$$u_{\widehat{T}}(\bm{\pi},\bm{p},\bm{r},s_1)\triangleq \mathbb{E}_{\mathbb{Q}(\bm{\pi})}\left[\sum_{t=1}^{\widehat{T}}\gamma^{t-1}r(s_t,a_t)+\gamma^{\widehat{T}}v_\infty\left(\tilde{s}_{\widehat{T}}\right)\right],$$
which stops at stage $\widehat{T}$ with a terminal reward $v_\infty(\cdot).$ Since $|\mathcal{S}|$ is finite and all $\mathcal{D}_n$, $n \in [N_s]$ and $\mathcal{W}$ are bounded, there exists a universal constant $r_{\max}\triangleq\max_{{s,a}}|r(s,a)|$ independent of $\widehat{T}$ such that for any $(\bm{\pi},\bm{p},\bm{r})$ where the uncertain parameters $(\bm{p},\bm{r})$ obey a joint probability distribution $\mathbb{P}\in\mathcal{F_S}$, the inequality 
$\left\vert u_{\widehat{T}}(\bm{\pi},\bm{p},\bm{r},s_1)-u(\bm{\pi},\bm{p},\bm{r},s_1)\right\vert \leq\gamma^{\widehat{T}}r_{\max} $ holds. This implies that for any $\mathbb{P}\in\bar{\mathcal{G}}^\infty_{\mathcal{S}}$,
\begin{equation}\label{eq:thm3-1-1}
\left\vert\int u_{\widehat{T}}(\bm{\pi},\bm{p},\bm{r},s_1){\rm d}\mathbb{P}(\bm{p},\bm{r})-\int u(\bm{\pi},\bm{p},\bm{r},s_1){\rm d}\mathbb{P}(\bm{p},\bm{r})\right\vert\leq\gamma^{\widehat{T}}r_{\max},
\end{equation}
and moreover,
\begin{equation}\label{eq:thm3-1}
\left\vert \inf_{\mathbb{P}\in\bar{\mathcal{G}}_{\mathcal{S}}^\infty}\int u_{\widehat{T}}(\bm{\pi},\bm{p},\bm{r},s_1){\rm d}\mathbb{P}(\bm{p},\bm{r})-\inf_{\mathbb{P'}\in\bar{\mathcal{G}}_{\mathcal{S}}^\infty}\int u(\bm{\pi},\bm{p},\bm{r},s_1){\rm d}\mathbb{P'}(\bm{p},\bm{r})\right\vert \leq\gamma^{\widehat{T}}r_{\max}.
\end{equation}
By \cref{theorem:S robust}, an $\mathcal{S}$-robust strategy $\bm{\pi}^*$ is a distributionally robust strategy of the finite horizon $\widehat{T}$ truncated problem for any $\widehat{T}\geq1$. Thus, we have
\begin{equation}\label{eq:thm3-2}
\inf_{\mathbb{P}\in\bar{\mathcal{G}}_{\mathcal{S}}^\infty}\int u_{\widehat{T}}(\bm{\pi}^*,\bm{p},\bm{r},s_1){\rm d}\mathbb{P}(\bm{p},\bm{r})\geq\inf_{\mathbb{P}'\in\bar{\mathcal{G}}_{\mathcal{S}}^\infty}\int u_{\widehat{T}}(\bm{\pi}',\bm{p},\bm{r},s_1){\rm d}\mathbb{P}'(\bm{p},\bm{r}) ~~\forall\pi'\in\Pi^{\text{HR}}.
\end{equation}
Combining \cref{eq:thm3-1} and \cref{eq:thm3-2} yields
\begin{equation*}
\begin{array}{ll}
&\inf\limits_{\mathbb{P}\in\bar{\mathcal{G}}_{\mathcal{S}}^\infty}\int u(\bm{\pi}^*,\bm{p},\bm{r},s_1){\rm d}\mathbb{P}(\bm{p},\bm{r})\\
\geq&\inf\limits_{\mathbb{P}'\in\bar{\mathcal{G}}_{\mathcal{S}}^\infty}\int u(\bm{\pi}',\bm{p},\bm{r},s_1){\rm d}\mathbb{P}'(\bm{p},\bm{r})-2\gamma^{\widehat{T}}r_{\max} ~~\forall\pi'\in\Pi^{\text{HR}}.
\end{array}
\end{equation*}
As the above inequality holds for an arbitrary $\widehat{T},$ we thus conclude that an $\mathcal{S}$-robust strategy $\bm{\pi}^*$ is also a distributionally robust strategy with respect to $\bar{\mathcal{G}}_{\mathcal{S}}$ of the infinite horizon distributionally robust MDP, which completes our proof. 
\end{proof}

\subsection{Stationary Model}
The stationary model treats the system as having a finite number of states, while multiple visits to one state is allowed. That is, if a state $s$ is visited for multiple times, then each time the distribution (of uncertain parameters) $\mathbb{P}_s$ is the same. For this model, it is much easier to develop statistically accurate sets of confidence when the underlying process is time invariant (see \cite{Nilim_Ghaoui_2005}). We define the stationary ambiguity set of admissible distributions as
$$\bar{\mathcal{G}}_\mathcal{S}\triangleq \left\{\mathbb{P}~\left\vert~\mathbb{P}=\bigotimes_{s\in \mathcal{S},t=1,2,\dots}\mathbb{P}_{s,t}, ~\mathbb{P}_{s,t}=\mathbb{P}_{s},~\mathbb{P}_{s}\in \mathcal{G}_s ~\forall s\in \mathcal{S}, ~t=1,2,\dots \right. \right\}.$$
Similar to the non-stationary model, we have the following result.
\begin{theorem}\label{thm4}
Under \cref{assu_mdp}, any $\mathcal{S}$-robust strategy is distributionally robust with respect to $\bar{\mathcal{G}}_{\mathcal{S}}$.
\end{theorem}

\begin{proof}[Proof]
Similar to the proof of \cref{thm3}, we consider the $\widehat{T}$ truncated problem. From the proof
of \cref{theorem:S robust}, for each $s\in\mathcal{S}$ and $t\leq\widehat{T}$, let $\bm{\pi}^*_{s,t}$ and $(\bm{p}^*_{s,t},\bm{r}^*_{s,t})$ be the optimal solution to the problem
$$
\max_{\bm{\pi}_s \in \mathcal{P}(\mathcal{A}_s)} \inf_{\mathbb{P} \in \mathcal{G}_s} \ep{\sum_{a \in \mathcal{A}_s} \pi_s(a) \left(\tilde{r}(s, a) + \gamma\sum_{s' \in \mathcal{S}} \tilde{p}(s' \vert s, a) v_{\infty}(s')\right)}
$$
and let $\mathbb{P}_{s,t}^*\in\mathcal{G}_s$ satisfy $\mathbb{E}_{\mathbb{P}_{s,t}^*}=(\bm{p}^*_{s,t},\bm{r}^*_{s,t})$. We have  $\bm{\pi}^*=\bigotimes_{s\in\mathcal{S},t\leq\widehat{T}}\pi^*_{s,t}$ and $\mathbb{P}^*=\bigotimes_{s\in\mathcal{S},t\leq\widehat{T}}\mathbb{P}^*_{s,t}$ and they satisfy
\begin{equation*}
\begin{array}{ll}
\displaystyle \max_{\pi\in\Pi^{\text{HR}}}\int u_{\widehat{T}}(\bm{\pi},\bm{p},\bm{r},s_1){\rm d}\mathbb{P}^*(\bm{p},\bm{r})&=\displaystyle\int u_{\widehat{T}}(\bm{\pi}^*,\bm{p},\bm{r},s_1){\rm d}\mathbb{P}^*(\bm{p},\bm{r})\\
&=\displaystyle\inf_{\mathbb{P}\in\bar{\mathcal{G}}_{\mathcal{S}}^\infty}\int u_{\widehat{T}}(\bm{\pi}^*,\bm{p},\bm{r},s_1){\rm d}\mathbb{P}(\bm{p},\bm{r}),
\end{array}
\end{equation*}
which leads to
\begin{equation*}
\begin{aligned}
\max_{\pi\in\Pi^{\text{HR}}}\int u(\bm{\pi},\bm{p},\bm{r},s_1){\rm d}\mathbb{P}^*(\bm{p},\bm{r})&\leq\inf_{\mathbb{P}\in\bar{\mathcal{G}}_{\mathcal{S}}^\infty}\int u(\bm{\pi}^*,\bm{p},\bm{r},s_1){\rm d}\mathbb{P}(\bm{p},\bm{r})+2\gamma^{\widehat{T}}r_{\max}\\
&\leq\inf_{\mathbb{P}\in\bar{\mathcal{G}}_{\mathcal{S}}}\int u(\bm{\pi}^*,\bm{p},\bm{r},s_1){\rm d}\mathbb{P}(\bm{p},\bm{r})+2\gamma^{\widehat{T}}r_{\max}.
\end{aligned}
\end{equation*}
Here the first inequality holds due to \cref{eq:thm3-1-1}, and the second inequality holds because $\bar{\mathcal{G}}_{\mathcal{S}}\subseteq\bar{\mathcal{G}}^\infty_{\mathcal{S}}$.

Note that by construction, $\bm{\pi}^*$ can be any $\mathcal{S}$-robust strategy. Moreover, $\bm{\pi}^*_{s,t}$ and $\bm{\mathbb{P}}^*_{s,t}$ are stationary, i.e., they do not depend on $t$. Hence, we have $\mathbb{P}^*\in\bar{\mathcal{G}}_{\mathcal{S}}$. Therefore,
\begin{equation}\label{eq:thm3-3}
\max_{\bm{\pi}\in \Pi^{\text{HR}}}\inf_{\mathbb{P}\in\bar{\mathcal{G}}_{\mathcal{S}}}\int u(\bm{\pi},\bm{p},\bm{r},s_1){\rm d}\mathbb{P}(\bm{p},\bm{r})\leq\max_{\bm{\pi}\in \Pi^{\text{HR}}}\int u(\bm{\pi},\bm{p},\bm{r},s_1){\rm d}\mathbb{P}^*(\bm{p},\bm{r}).
\end{equation}
Combining~\cref{eq:thm3-2} and~\cref{eq:thm3-3} leads to
\begin{equation*}
\max_{\bm{\pi}\in \Pi^{\text{HR}}}\inf_{\mathbb{P}\in\bar{\mathcal{G}}_{\mathcal{S}}}\int u(\bm{\pi},\bm{p},\bm{r},s_1){\rm d}\mathbb{P}(\bm{p},\bm{r})\leq\inf_{\mathbb{P}\in\bar{\mathcal{G}}_{\mathcal{S}}}\int u(\bm{\pi}^*,\bm{p},\bm{r},s_1){\rm d}\mathbb{P}(\bm{p},\bm{r})+2\gamma^{\widehat{T}}r_{\max}.
\end{equation*}
Since $\widehat{T}$ can be arbitrarily large, this concludes that any $\mathcal{S}$-robust strategy $\bm{\pi}^*$ is also a distributionally robust strategy with respect to $\bar{\mathcal{G}}_{\mathcal{S}}$. 
\end{proof}

\begin{remark}
The worst-case expected performance of the non-stationary model provides a lower bound to that of the stationary model since $\bar{\mathcal{G}}_\mathcal{S}\subseteq\bar{\mathcal{G}}^\infty_{\mathcal{S}}$. Therefore, we can use the non-stationary model to approximate the stationary model, when the latter is intractable in the finite horizon case. When horizon approaches infinity, such approximation becomes exact as shown in the proofs of \cref{thm3} and \cref{thm4}. In particular, the optimal solutions to both formulations coincide and they can be computed by iteratively solving a minimax problem. 
\end{remark}

\section{Numerical Experiments}\label{section6}
We consider a dynamic newsvendor problem over the finite decision horizon $t=1,\dots, T.$ At each period $t$, the newsvendor observes the current level of stock and makes an ordering decision to replenish the stock before the random demand is realized. For each period $ t $, let $d_t$ be the random demand observed, $a_t$ be the order quantity placed, and $s_t$ be the inventory position at the beginning of that period (negative $s_t$ is understood as backorder). There are limits on the inventory position---that is, $s_t \in [s_{\min},s_{\max}]$, where $-s_{\min}$ is the largest allowable backorder and $s_{\max}$ is the maximal number of units we can keep. The inventory position evolves according to the dynamics:
\begin{equation*}
s_{t+1}=\phi(s_t+a_t-d_t) ~~\forall t=1,\dots,T-1,
\end{equation*} 
where the initial inventory $ s_1 = 0 $ and 
$$
\phi(s) = \left\{
\begin{array}{ll}
\max\{s, s_{min}\} &~{\rm if~} s \leq 0\\
\min\{s,s_{\max}\} &~{\rm otherwise.}
\end{array}
\right.
$$
Let $c_t$, $h_t$ and $b_t$ be the unit order, holding, and backorder cost at time $t=1,\dots, T $, respectively. Given the inventory position $ s_t $ and the ordering decision $a_t$, the cost in each period $t$ is $c_t a_t+\max\{h_t s_t,-b_t s_t\}$
and the cost in the terminal period $T$ is $\max\{h_Ts_T, -b_Ts_T\}.$ Note that the periodic cost only depends on the current state $ s_t $ and the current ordering decision $a_t$.

We assume that the realizations of demand $\{d_t\}_{t=1}^{T-1}$ and actions $\{a_t\}_{t=1}^{T-1}$ are all integer-valued. Hence we can formulate the problem as a finite-state and finite-action MDP, where the 
state space is given by $\mathcal{S} \triangleq \{s_{\min}, s_{\min}+1, \dots, s_{\max}-1, s_{\max} \}$ and the set of admissible actions in state $s$ is given by $\mathcal{A}_s = \{0,1,\dots, s_{\max} - s\} $. We assume for ease of exposition that the true distribution $ \mathbb{P}_0 $ of the random demand is stationary such that $ \mathbb{P}_0[\tilde{d}_t = i] = p_i$. Solving the inventory problem using the classical dynamic programming approach requires the the perfect knowledge of the distribution of $\tilde{d}$, i.e., known $ p_{i} $. Such information is usually hard, if not impossible, to obtain in practice (see e.g., \cite{Bertsekas_2000, Sutton_Barto_2011}). Instead, we adopt the distributionally robust approach, which is favorable when we only have sparse observations of $p_i$. We assume that $ \bm{p} = (p_i)_i $ is uncertain and the distribution of $ \bm{p} $ is ambiguous but belongs to an ambiguity set. By \cref{def}, we compute the $\mathcal{S}$-robust order strategy for all $s\in \mathcal{S} $ and $ t = 1,\dots, T-1$ via solving
$$
\begin{array}{l}
v_t(s) = \\  \min\limits_{\bm{\pi}_s \in \mathcal{P}(\mathcal{A}_s)}\sup\limits_{\mathbb{P} \in \mathcal{G}_s} \mathbb{E}_\mathbb{P}\bigg[\sum\limits_{a_t \in \mathcal{A}_s}\pi_s(a_t)(c_t a_t+\max\{h_ts,-b_ts\}+\sum\limits_{i} v_{t+1}(\phi(s+a_t-i))\tilde{p}_i)\bigg].
\end{array}
$$
Here $ v_T(s) = \max\{h_Ts,-b_Ts\} $ and for all $ s $, $ \mathcal{G}_s$ is the same ambiguity set of $\bm{p}$.  

Particularly, we consider a more practical data-driven setting described as follows. The decision maker initially has access only to  $ N $ independent training samples $ \bm{p}_1^\dag, \dots, \bm{p}_N^\dag $ of the true data-generating distribution $\bm{p}$. The decision maker then obtains the empirical distribution $\mathbb{P}^\dag = \frac{1}{N}\sum_{i \in [N]}\delta_{\bm{p}^\dag_i} $ and afterwards, she solves the $ \mathcal{S}$-robust order strategy under the Wasserstein ambiguity set centered around $\mathbb{P}^\dag$ with different values of radius $\theta$. Another testing dataset that consists of $10000$ independent samples of $\{d_t\}_{t=1}^{T-1}$ is then generated from $ \bm{p} $ and the total cost of applying the $\mathcal{S}$-robust order strategy is simulated by averaging over these $10000$ runs. The out-of-sample performance looks at this stimulated total cost and is reported by running the experiment $2000$ times--- each time using a new generated training and testing datasets. The parameters we used are as follows: $ T = 5$, $ s_{\min} = -5 $, $ s_{\max} = 10 $, $ c_t = 1 $, $ h_t = 2 $, $ b_t = 3 $, $ \tilde{d} \in \{0,1,2,3,4\} $, and $ \bm{p} = (0.05,0.4,0.1,0.4,0.05) $ such that the true demand distribution has two modes at $ 1 $ and $ 3 $ and has an expectation of $ 2 $.

\begin{figure}[htbp]
\begin{subfigure}{.5\textwidth}
\begin{center}
\includegraphics[width=1\linewidth]{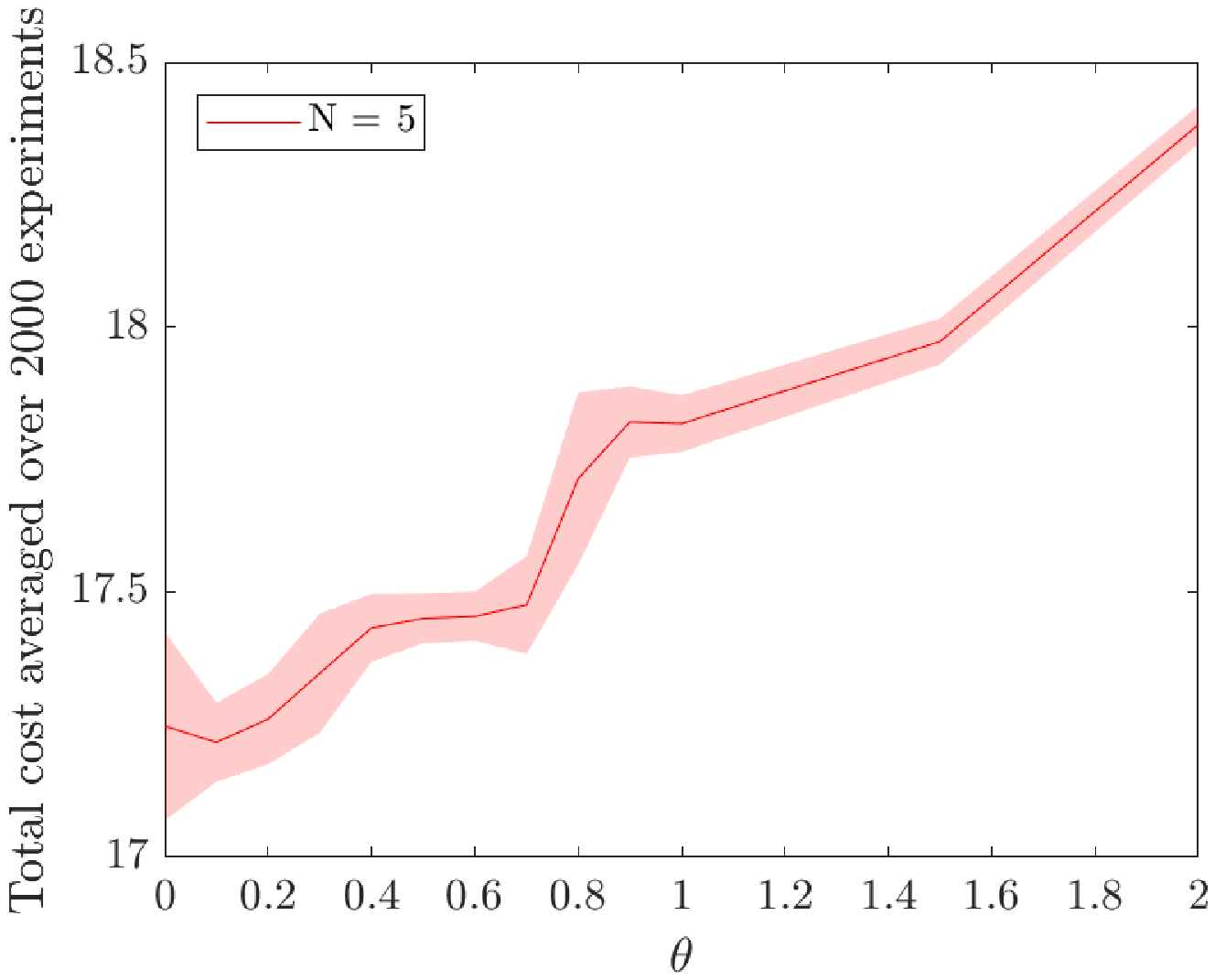}
\end{center}
\end{subfigure}%
\begin{subfigure}{.5\textwidth}
\begin{center}
\includegraphics[width=1\linewidth]{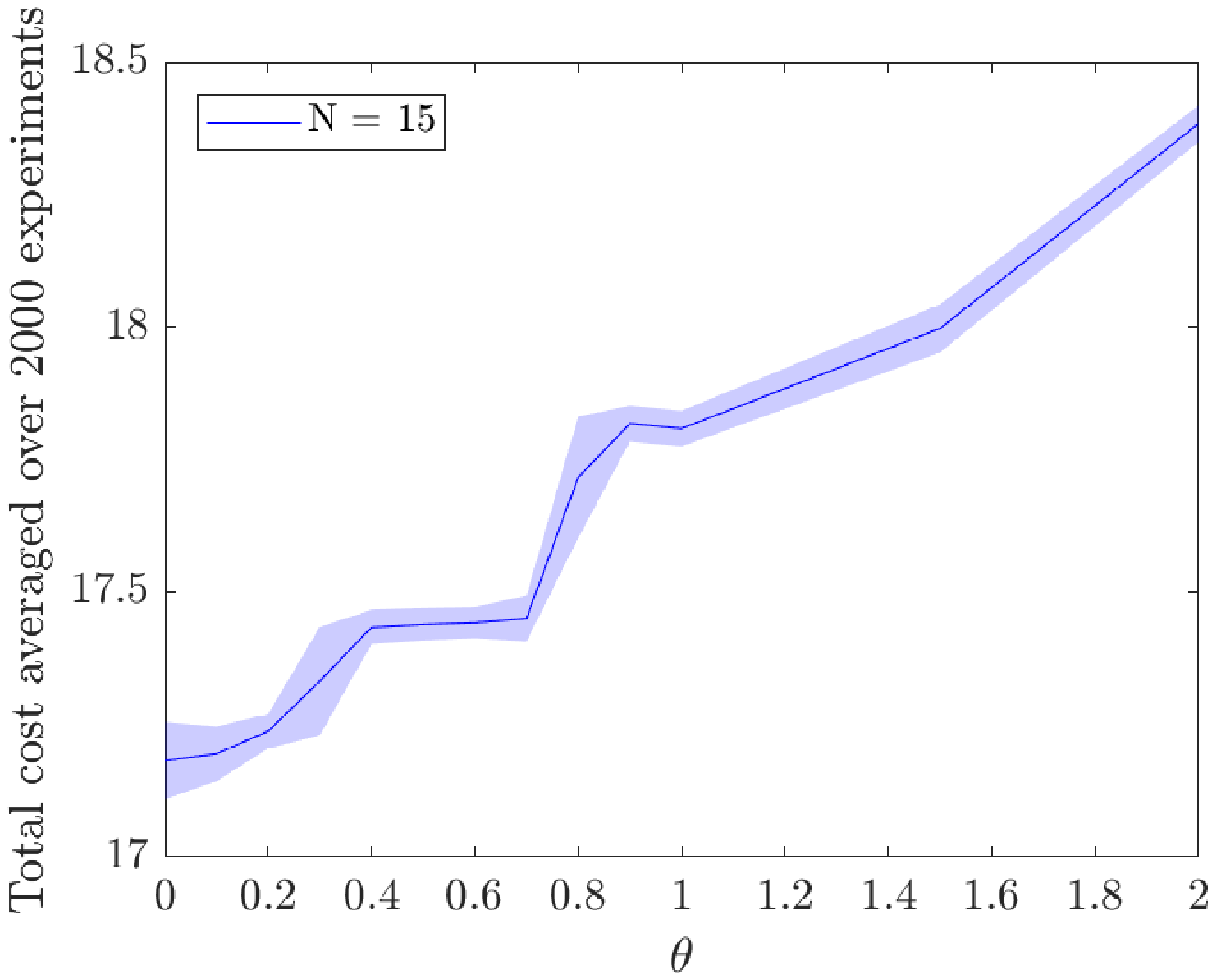}
\end{center}
\end{subfigure}
\vspace{-10pt}
\caption{{\textnormal{Out-of-sample performance for various Wasserstein radii. The shaded regions cover range from `mean minus one standard deviation' to `mean plus one standard deviation' of $2000$ randomly generated experiments, whereas the solid lines describe the mean statistics.}} \label{figure::errorbars}}
\end{figure}

In \cref{figure::errorbars}, we report the distribution of the total cost for different Wasserstein radii over $2000$ experiments. Note that when $ \theta = 0 $, the Wasserstein approach recovers the sample average approach that considers the empirical MDP using $\mathbb{P}^\dag$ (see e.g., \cite{Haskell_Jain_Kalathil_2016}); when $ \theta \geq 2 $, the Wasserstein approach recovers the classical robust optimization approach (see e.g., \cite{Iyengar_2005, Nilim_Ghaoui_2005}). In general, a larger size of training samples leads to a better out-of-sample performance in terms of variability, see the comparison between $N=5$ and $N=15$. For most values of the radius, a larger size of training samples yields a better out-of-sample performance that has a slightly lower mean value and a significant lower standard deviation of the total cost averaged over $2000$ experiments. As the radius $\theta$ increases, the mean of the out-of-sample performance increases because the solution becomes more conservative while the standard deviation fluctuates. Recently, the work \cite{Gotoh_Kim_Lim_2015} studies the effect of the radius $\theta$ on the solution obtained from considering the $\phi$-divergence ambiguity set---another popular statistical-distance-based ambiguity set. The authors show the existence of a mean-standard deviation trade-off such that when the radius $\theta$ increases, the mean increases but the standard deviation would decrease. However, we do not observe such a trade-off in our experiment that considers the Wasserstein ambiguity set as some $\theta$ finds itself at less of an advantage than others (because of a higher mean and a larger standard deviation, see for example in the case of $N = 15$, $\theta \in \{0.3, 0.7, 0.8\}$ seems to be dominated by $\theta = 0.2 $). Therefore, it begs an interesting question as how to calibrate a good radius~$\theta$ for the Wasserstein ambiguity set.

\section{Conclusion}\label{section7}
In this paper, we address MDPs under parameter uncertainty following the distributionally robust approach, to mitigate the conservatism of the classical robust MDP framework which only considers the set-inclusive formulation of uncertainty, and to incorporate additional a priori probabilistic information regarding the unknown parameters. Specifically, we generalize existing works on distributionally robust MDPs by investigating a unified format of ambiguity sets that provides extra modeling power. We hope our analysis based on such a unified format would encourage study of distributionally robust MDPs with new types of ambiguity sets, including (i) hybridizations of generalize-moment-based and statistical-distance-based ambiguity sets and (ii) mixture-distribution-based ambiguity sets. Our solution approach leads to an optimal strategy that can be obtained through a Bellman type backward induction by solving a sequence of classical robust optimization problems, which can now be effectively modeled and solved by using algebraic modeling packages.

In our numerical studies, we present the out-of-sample performance of Wasserstein ambiguity sets with different radii. Calibrating the radius $\theta$ shall be an interesting problem---moving forward, we intend to apply resampling approaches such as cross-validation and bootstrap to calibrate $\theta$. Inspired by recent works in incorporating dependence structure about the uncertainty (see e.g., \cite{Gao_Kleywegt_2017_a, Gao_Kleywegt_2017_b}), we are also interested in integrating that information with the Wasserstein ambiguity set and study the performance in distributionally robust MDPs. In addition, we wish to adapt approximate dynamic programming techniques (see e.g., \cite{Powell_2007}), and develop scalable methods for distributionally robust MDPs with large or even continuous state and action spaces. 


\bibliographystyle{siamplain}
\bibliography{references}
\end{document}